\documentclass[reqno,11pt]{amsart}

\usepackage{xcolor}
\usepackage{graphicx}
\usepackage[hidelinks]{hyperref}
\usepackage{amscd}

\usepackage{amsmath}
\usepackage{amsthm}
\usepackage{amssymb}
\usepackage{latexsym,array}
\usepackage{amsfonts}
\usepackage{shadow}
\usepackage{amsbsy}
\usepackage{amssymb}
\usepackage{dsfont}
\usepackage{mathtools}
\usepackage{enumitem}
\usepackage[notref, notcite, final]{showkeys}
\usepackage{color}
\usepackage{graphicx}
\usepackage[hidelinks]{hyperref}

\usepackage{dsfont}

\newtheorem{theorem}{Theorem}[section]
\newtheorem{definition}[theorem]{Definition}

\newtheorem{lemma}[theorem]{Lemma}
\newtheorem{proposition}[theorem]{Proposition}
\newtheorem{corollary}[theorem]{Corollary}
\newtheorem{remark}[theorem]{Remark}

\newcommand{\rr}{\mathbb{R}}
\newcommand{\cc}{\mathbb{C}}

\newcommand{\nn}{\mathbb{N}}

\newcommand{\unx}{\underline{x}}
\newcommand{\uny}{\underline{y}}

\newcommand{\bx}{\mathbf{x}}

\title[Superoscillations in the hypercomplex setting]{Superoscillations in the hypercomplex setting}

\author[F. Colombo]{F. Colombo}
\address{(FC) Politecnico di Milano, Dipartimento di Matematica, Via E. Bonardi 9, 20133 Milano, Italy}
\email{fabrizio.colombo@polimi.it}

\author[F. Mantovani]{F. Mantovani}
\address{(FM) Politecnico di Milano, Dipartimento di Matematica, Via E. Bonardi 9, 20133 Milano, Italy}
\email{francesco.mantovani@polimi.it}

\author[S. Pinton]{S. Pinton}
\address{(SP) Politecnico di Milano, Dipartimento di Matematica, Via E. Bonardi 9, 20133 Milano, Italy}
\email{stefano.pinton@polimi.it}

\author[P. Schlosser]{P. Schlosser}
\address{(PS) Graz University of Technology, Institut f\"ur Angewandte Mathematik, Steyrergasse 30, 8010 Graz, Austria}
\email{pschlosser@math.tugraz.at}

\begin{document}

\begin{abstract}
Superoscillatory functions represent a counterintuitive phenomenon in physics but also in mathematics,
where a band-limited function oscillates faster than its highest Fourier component.
They appear in various contexts, including quantum mechanics, as a result of a weak
measurement introduced by Y. Aharonov and collaborators.
These functions can be extended to the complex variable and are a specific instance of the more general notion of supershift.
The aim of this paper is to extend the notion of superoscillatory functions to the hypercomplex setting.
This extension is richer than the complex case since the Fueter-Sce extension theorem for Clifford-valued functions provides two notions of hyperholomorphic functions.
We will explore these notions and address the corresponding superoscillating theories.
 \end{abstract}
 \maketitle


\vskip 1cm

\textbf{Acknowledgements:}
F. Colombo and S. Pinton are supported by MUR grant "Dipartimento di Eccellenza 2023-2027".
 P. Schlosser was funded by the Austrian Science Fund (FWF) under Grant No. J 4685-N and by the European Union -- NextGenerationEU.

\tableofcontents

\section{Introduction }
Superoscillations are a  somewhat counterintuitive phenomenon for which a band-limited function
locally oscillates faster than its highest Fourier component.
This means that for example a signal can exhibit rapid variations in certain regions,
even though its overall frequency content is limited to lower frequencies.
The problem of constructing a sequence of superoscillatory functions can be described as follows:
Let $h_j(n)$ be a given set of points in $[-1,1]$, $j=0,1,...,n$, for $n\in \mathbb{N}$ and let $a\in \mathbb{R}$ be such that $|a|>1$.
Determine the coefficients $X_j(n,a)$ of the sequence
$$
f_n(x)=\sum_{j=0}^nX_j(n,a)e^{ih_j(n)x},\ \ \ x\in \mathbb{R}
$$
in such a way that
$$
f_n^{(p)}(0)=(ia)^p,\ \ \ {\rm for} \ \ \ p=0,1,...,n.
$$
The conditions $f_n^{(p)}(0)=(ia)^p$  mean that the functions $x\mapsto e^{iax}$ and $f_n(x)$ have the same derivatives at the origin, for
$p=0,1,...,n$,  so they have the same Taylor polynomial of order $n$.
Under the condition that, for every fixed $n\in\mathbb{N}$, the points  $h_j(n)$  for $j=0,...,n$, (often denoted by $h_j$) are distinct we
obtain an explicit formula  for the coefficients  $X_j(n,a)$ that are given by
$$
X_j(n,a)=
\prod_{k=0,\  k\not=j}^n\Big(\frac{h_k(n)-a}{h_k(n)-h_j(n)}\Big),
$$
so the superoscillating sequence $f_n(x)$ takes  the explicit  form
$$
f_n(x)=\sum_{j=0}^n\prod_{k=0,\  k\not=j}^n\Big(\frac{h_k(n)-a}{h_k(n)-h_j(n)}\Big)\ e^{ih_j(n)x},\ \ \ x\in \mathbb{R}.
$$
Observe that, by construction, this function is band limited and it converges to $e^{iax}$ with arbitrary $|a|>1$, so it is superoscillating.
These functions where introduced in \cite{newmet} and further studied in \cite{PeterTAMS}.

The prototypical superoscillating function, obtained in a different way with respect to what we have discussed above, is
associated with the weak values in quantum mechanics \cite{aav}, and is defined as
\begin{equation}\label{FNEXP}
F_n(x,a)=\sum_{j=0}^nC_j(n,a)e^{i(1-2j/n)x},\ \ x\in \mathbb{R},
\end{equation}
where $a>1$ and the coefficients $C_j(n,a)$ are given by
\begin{equation}\label{Ckna}
C_j(n,a)={n\choose j}\left(\frac{1+a}{2}\right)^{n-j}\left(\frac{1-a}{2}\right)^j.
\end{equation}
If we fix $x \in \mathbb{R}$  and we let $n$ go to infinity, we  obtain that
$$
\lim_{n \to \infty} F_n(x,a)=e^{iax}.
$$
Clearly the name superoscillations comes from the fact that
in the Fourier's representation of the function (\ref{FNEXP}) the frequencies
$1-2j/n$ are bounded by 1, but the limit function $e^{iax}$ has a frequency $a$ that can be arbitrarily larger  than $1$.
A precise definition of superoscillating functions is as follows.

\medskip
We call {\em generalized Fourier sequence}
a sequence of the form
\begin{equation}\label{basic_sequence}
f_n(x):= \sum_{j=0}^n X_j(n,a)e^{ih_j(n)x},\ \ \ n\in \mathbb{N},\ \ \ x\in \mathbb{R},
\end{equation}
where $a\in\mathbb R$, $X_j(n,a)\in\mathbb{C}$ and $h_j(n)\in\mathbb{R}$.
A generalized Fourier sequence of the form (\ref{basic_sequence})
 is said to be {\em a superoscillating sequence} if
 $h_j(n)\in[-1,1]$ and
 there exists a compact subset of $\mathbb R$,
 which will be called {\em a superoscillation set},
 on which $f_n(x)$ converges uniformly to $e^{ig(a)x}$,
 where $g$ is a continuous real valued function such that $|g(a)|>1$.

The evolution of superoscillations have been widely studied in the literature and
many investigations are going on also nowadays by different research groups
in mathematics and in physics. See, without claiming completeness,
\cite{ABCS19}-\cite{acsst5} and \cite{ACJSSST22}-\cite{uno},\cite{Jussi}, \cite{Be19}-\cite{b4}, \cite{Talbot}, \cite{due},\cite{Pozzi}.
Several results on the evolution of superoscillations are based on continuity properties of
infinite order differential operators \cite{AOKI,aoki3}.
Infinite order differential  operators in the hypercomplex setting have been recently considered in
\cite{CKPS24}, \cite{CPSS21},  \cite{EDIM_PPSS} for the entire hyperholomorphic function settings see also the paper \cite{CAK07}.
Supershift and superoscillations and the case of several variables
\cite{hyper}-\cite{CSSYgenfun}.
Recently in the literature appeared and interesting relations of superoscillations associated with Fock spaces \cite{DIKI},
to the short-time Fourier transform \cite{SHORTTIME}, and to
some applications to technology \cite{kempf1,kempf2HHH}.

\medskip
We highlight an important aspect in the investigation of superoscillations.
Extending superoscillatory functions to the complex setting is particularly useful,
as this extension to entire functions allows the use of infinite order differential operators,
which play a crucial role in studying the evolution problem of superoscillations
by Schr\"odinger equation or by any other field equation in quantum mechanics.
 \medskip

The aim of this article is to describe the phenomenon of superoscillations also in the setting of the Clifford algebra $\mathbb{R}_n$. One important feature of $\mathbb{R}_n$ is that there are several (natural) notions of holomorphicity. The most important ones are the {\it slice hyperholomorphic functions} and the {\it monogenic functions}. Note, that these two function spaces are connected by the Fueter-Sce extension theorem, see \cite{ColSabStrupSce}. Thus, the challenge is to define what we mean by superoscillations in these two settings.

\medskip
{\em The content of the paper.}
{\color{black} The main object of study of this paper are Clifford valued superoscillating functions defined on paravectors that is a subset of the Clifford-Algebra. We adopt two different approaches to define this kind of sequences. The first one treats the Clifford variable in a unified way, in particular we give the following definition:
\begin{definition}[Superoscillating sequence.]\label{SUPOSONE}
Let $e_j$, for $j=1,\ldots,n$ be the imaginary units of the Clifford algebra $\mathbb{R}_n$ with the usual defining properties.
A sequence of the form
\begin{equation}\label{basic_sequenceq}
f_N(x_1,\dots, x_n;a):= \sum_{j=0}^N Z_j(N,a)e^{h_j(N)(e_1x_1+\dots +e_nx_n)},\ \ \ n\in \mathbb{N},\ \ \ (x_1,\dots ,x_n)\in \mathbb{R}^n,
\end{equation}
where $a\in\mathbb{R}$, $Z_j(N,a)\in\mathbb{R}_n$, $h_j(n)\in[-1,1]$, is said to be {\em a superoscillating sequence} if
 $\sup_{j,n}|h_j(N)|\leq 1$ and
 $f_N(x_1,\dots,x_n,a)$ converges uniformly in $\unx=e_1x_1+\dots +e_nx_n$ on any compact subset of $\mathbb R^{n}$ to $e^{g(a)\unx}$,
 where $g:\mathbb{R}\rightarrow\mathbb{R}$ is a continuous such that $|g(a)|>1$.
\end{definition}
\begin{remark}\label{remark1}
In the Definition \ref{SUPOSONE}, we decided to put the coefficients $Z_j(N,a)$ on the left of the exponential functions,
 but we could give analogous definitions putting $Z_j(N,a)$ on the right of the exponential functions.
 This remark is important, as we will see, when the left or right regularity
 of the  extension of a superoscillating sequence has to be taken into account.
\end{remark}
Note that in the Definition \ref{SUPOSONE} the role of $e_1x_1+\dots +e_nx_n$ is the same of the role of $ix$ in the definition of the superoscillating sequences in the complex setting. We prove the existence of the archetypical superoscillating sequence and its extension to a sequence of entire slice hyperholomorphic  functions  or to a sequence of entire monogenic functions that converges in the corresponding spaces of entire exponentially increasing functions (see Lemma \ref{l1} and \ref{l2}).

The second approach is more flexible and it admits the possibility in the definition of a superoscillating sequence to multiply the
 variables: $x_1,\dots, x_n$ with different frequencies  (see Definition \ref{superoscill} for the slice hyperholomorphic setting and Definition \ref{superoscill2} for the monogenic setting). In particular, in Corollary \ref{cor1} (resp. Corollary \ref{cor2}) we prove the construction of a superoscillating sequence according to Definition \ref{superoscill} (resp Definition \ref{superoscill2}) starting from a superoscillating sequence as in Definition \ref{SUPOSONE} that admits an extension to a sequence of entire slice hyperholomorphic functions (resp. monogenic functions) converging in the spaces of entire slice hyperholomorphic functions (resp. entire monogenic functions).

Finally we also define, as a generalization of the notion of superoscillating sequence, the notion of supershift sequences of functions of Clifford variable in the slice hyperholomorphic setting (see Definition \ref{d4} and Definition \ref{ss1}) and in the monogenic setting (see Definition \ref{d5} and Definition \ref{ss2}) and in Theorem \ref{propm=3q} (resp. Theorem \ref{propm=5q}) we prove that starting from a superoscillating sequence it is possible to construct a supershift sequence according to Definition \ref{ss1} (resp. Definition \ref{ss2}). }

\section{Preliminary material on hyperholomorphic functions}

In this section we recall some results on slice hyperholomorphic functions (see Chapter 2 in \cite{CSS11})
and we prove some important properties of entire slice hyperholomorphic and entire monogenic functions that appear here for the first time.
We recall that $\rr_n$ is the real Clifford algebra over $n$ imaginary units $e_1,\ldots ,e_n$.
The element
 $(x_0,x_1,\ldots,x_n)\in \mathbb{R}^{n+1}$ will be identified with the paravector
$
 x=x_0+\underline{x}=x_0+ \sum_{\ell=1}^nx_\ell e_\ell
$
and the real part $x_0$ of $x$ will also be denoted by ${\rm Re}(x)$.
An element in $\mathbb{R}_{n}$, called a {\em Clifford number}, can be written as
$$
a=a_0+a_1 e_1+\ldots +a_ne_n+a_{12}e_1e_2+\ldots+a_{123}e_1e_2e_3+\ldots+a_{12...n}e_1e_2...e_n.
$$
Denote by $A=\{i_1,\dots, i_r\}$ with $0\leq r\leq n$,  an element in the power set $P(1,\ldots ,n)$ with ordered indexes $i_{1}<i_2<\dots< i_r$.
The element $e_{i_1}\ldots e_{i_r}$ can be written as $e_{i_1...i_r}$ or, in short, $e_A$.
Thus, in a more compact form, we can write a Clifford number as
$$
a=\sum_Aa_A e_A.
$$
When $A=\emptyset$ we set $e_\emptyset=1$.
The Euclidean norm of an element $y\in \mathbb{R}_n$ is
given by $|y|^2=\sum_{A} |y_A|^2$,
in particular the norm of the paravector $x\in\mathbb{R}^{n+1}$ is $|x|^2=x_0^2+x_1^2+\ldots +x_n^2$.
The conjugate of $x$ is given by
$
\bar x=x_0-\underline x=x_0- \sum_{\ell=1}^nx_\ell e_\ell.
$
Recall that $\mathbb{S}$ is the sphere
$$
\mathbb{S}=\{ \underline{x}=e_1x_1+\ldots +e_nx_n\ | \  x_1^2+\ldots +x_n^2=1\};
$$
so for $\mathbf{j}\in\mathbb{S}$ we have $\mathbf{j}^2=-1$.
Given an element $x=x_0+\underline{x}\in\rr^{n+1}$ let us define
$
\mathbf{j}_x=\underline{x}/|\underline{x}|$ if $\underline{x}\not=0,
$
 and given an element $x\in\rr^{n+1}$, the set
$$
[x]:=\{y\in\rr^{n+1}\ :\ y=x_0+{\mathbf{j}} |\underline{x}|, \ \mathbf{j}\in \mathbb{S}\}
$$
which is an $(n-1)$-dimensional sphere in $\mathbb{R}^{n+1}$.
The vector space $\mathbb{R}+\mathbf{j}\mathbb{R}$ will be denoted by $\mathbb{C}_\mathbf{j}$ and
an element belonging to $\mathbb{C}_\mathbf{j}$ will be indicated by $u+\mathbf{j}v$, for $u$, $v\in \mathbb{R}$.
With an abuse of notation we will write $x\in\mathbb{R}^{n+1}$.
Thus, if $U\subseteq\mathbb{R}^{n+1}$ is an open set,
a function $f:\ U\subseteq \mathbb{R}^{n+1}\to\mathbb{R}_n$ can be interpreted as
a function of the paravector $x$.
We say that $U \subseteq \mathbb{R}^{n+1}$ is axially symmetric if $[x]\subset U$  for any $x \in U$.

\subsection{Slice hyperholomorphic functions}
This is the first class of hyperholomorphic function  due to the  Fueter-Sce extension theorem, see \cite{ColSabStrupSce},
associated with the second map $T_{FS1}$

\begin{definition}[Slice hyperholomorphic functions]\label{SHolDefMON}
 Let $U\subseteq \mathbb{R}^{n+1}$ be an axially symmetric open set and let
 $\mathcal{U} = \{ (u,v)\in\rr^2: u+ \mathbb{S} v\subset U\}$. A function $f:U\to \mathbb{R}_n$ is called a left slice function, if it is of the form
 \[
 f(q) = f_{0}(u,v) + \mathbf{j}f_{1}(u,v)\qquad \text{for } q = u + \mathbf{j} v\in U
 \]
where the two functions $f_{0},f_{1}: \mathcal{U} \to \mathbb{R}_n$ satisfy the compatibility conditions
\begin{equation}\label{CCondmon}
f_{0}(u,-v) = f_{0}(u,v),\qquad f_{1}(u,-v) = -f_{1}(u,v).
\end{equation}
If in addition $f_{0}$ and $f_{1}$ are $C^1$ regular and satisfy the Cauchy-Riemann-equations
 \begin{align}\label{CRMMON}
\frac{\partial}{\partial u} f_{0}(u,v) - \frac{\partial}{\partial v} f_{1}(u,v) &= 0\\
\frac{\partial}{\partial v} f_{0}(u,v)+ \frac{\partial}{\partial u} f_{1}(u,v) &= 0,
\end{align}
 then $f$ is called left slice hyperholomorphic.
A function $f:U\to \mathbb{R}_n$ is called a right slice function if it is of the form
\[
f(q) = f_{0}(u,v) + f_{1}(u,v) \mathbf{j}\qquad \text{for } q = u+ \mathbf{j}v \in U
\]
with two functions $f_{0},f_{1}: \mathcal{U} \to \mathbb{R}_n$ that satisfy \eqref{CCondmon}.
If in addition $f_{0}$ and $f_{1}$ satisfy the Cauchy-Riemann-equation, then $f$ is called right slice hyperholomorphic.
\end{definition}

If $f$ is a left (or right) slice function such that $f_{0}$ and $f_{1}$ are real-valued, then $f$ is called intrinsic.
We denote the sets of left (resp. right) slice hyperholomorphic functions on $U$ by $\mathcal{S\!M}_L(U)$
(resp. $\mathcal{S\!M}_R(U)$).

\begin{definition}\label{slicederivative}
Let $f: U\subseteq\mathbb{R}^{n+1}\to\mathbb{R}_n$ and let $x = u + \mathbf{j}v\in U$. If $x$ is not real, then we say that $f$ admits left slice derivative in $x$ if
\begin{equation}\label{SDerivL}
\partial_S f(x) := \lim_{p\to x, \,  p\in\cc_{\mathbf{j}}} (p-x)^{-1}(f(p)-f(x))
\end{equation}
exists and is finite. If $x$ is real, then we say that $f$ admits left slice derivative in $x$ if \eqref{SDerivL} exists for any $\mathbf{j}\in\mathbb{S}$.
Similarly, we say that $f$ admits right slice derivative in a nonreal point $x = u + \mathbf{j}v \in U$ if
\begin{equation}\label{SDerivR}
\partial_S f(x) := \lim_{p\to x, \,  p\in\cc_{\mathbf{j}}}(f(p)-f(x))(p-x)^{-1}
\end{equation}
exists and is finite, and we say that  $f$ admits right slice derivative in a real point
$x\in U$ if \eqref{SDerivR} exists and is finite, for any $\mathbf{j}\in\mathbb{S}$.
\end{definition}
\begin{remark}
Observe that $\partial_S f(x)$ is uniquely defined and independent of the choice of $\mathbf{j}\in\mathbb{S}$ even if $x$ is real.
If $f$ admits a slice derivative in one point $x\in\mathbb R^{n+1}$, then $f_{\mathbf{j}}$ is $\cc_\mathbf{j}$-complex left, resp. right, differentiable in $x$ and we find
\begin{equation}\label{SDerivPartial}
\partial_S f(x) = f_{\mathbf{j}}'(x) = \frac{\partial}{\partial u} f_{\mathbf{j}}(x) = \frac{\partial}{\partial u}f(x),\quad x = u + \mathbf{j}v.
\end{equation}
\end{remark}

\begin{theorem}[\cite{CSS11}, Proposition 2.3.1]\label{PowSerThm}
Let $a\in\mathbb{R}$, $r>0$ and $B_{r}(a) = \{x\in\mathbb R^{n+1}: |x-a|<r\}$. If $f\in\mathcal{S\!M}_L(B_r(a))$, then
\begin{equation}
\label{PowSerL}
f(x) = \sum_{k= 0}^{+\infty} (x-a)^k\frac{1}{k!} \partial_S^k f(a)\qquad \forall x = u + \mathbf{j} v \in B_r(a).
\end{equation}
If on the other hand $f\in\mathcal{S\!M}_R(B_r(a))$, then
\[
f(x) = \sum_{k= 0}^{+\infty}\frac{1}{k!} \partial_S^k f(a)  (x-a)^k\qquad \forall x = u + \mathbf{j} v \in B_r(a).
\]
\end{theorem}

We now recall the natural product that preserves slice hyperholomorphicity of functions admitting power series expansion
as shown by Theorem \ref{PowSerThm}.
\begin{definition}
Let  $f(x) = \sum_{k=0}^{+\infty}x^{k}a_k$ and
$g(x) = \sum_{k=0}^{+\infty}x^kb_k$, defined on a disc $B_r(0)$, be two left slice hyperholomorphic power series,
the left-star product is defined by
\begin{equation}\label{ProdLSeries}
  (f\star_L g) (x) = \sum_{\ell=0}^{+\infty} x^\ell \left(\sum_{k=0}^{\ell}a_{k}b_{\ell-k}\right),
\end{equation}
and leads again to a left slice hyperholomorphic function on $B_r(0)$.
Similarly, for right slice hyperholomoprhic power series
$
f(x)=\sum_{k=0}^{+\infty}a_kx^{k}
$ and
$
g(x)=\sum_{k=0}^{+\infty}b_kx^{k}
$
the right-star product is defined by
\begin{equation}\label{ProdRSeries}
 (f\star_R g) (x) =  \sum_{\ell=0}^{+\infty} \left( \sum_{k=0}^{\ell}a_{k}b_{\ell-k}\right)x^\ell.
\end{equation}
\end{definition}
The Cauchy formula of slice hyperholomoprhic functions has two different Cauchy kernels
according to left or right slice hyperholomoprhicity.
Let $x,s\in \mathbb{R}^{n+1}$, with $x\not\in [s]$,
be paravectors then the slice hyperholomorphic Cauchy kernels are defined by
$$
S_L^{-1}(s,x):=-(x^2 -2 {\rm Re}  (s) x+|s|^2)^{-1}(x-\overline s),
$$
and
$$
S_R^{-1}(s,x):=-(x-\bar s)(x^2-2{\rm Re} (s)x+|s|^2)^{-1}.
$$
The two results below can be found in \cite{CSS11}, Section 2.8.
\begin{theorem}[The Cauchy formulas for slice hyperholomoprhic functions]
\label{CauchygeneraleMONOG}
Let $U\subset\mathbb{R}^{n+1}$ be an axially symmetric domain.
Suppose that $\partial (U\cap \mathbb{C}_\mathbf{j})$ is a finite union of
continuously differentiable Jordan curves  for every $\mathbf{j}\in\mathbb{S}$ and set
 $ds_\mathbf{j}=-ds \mathbf{j}$.
 Let $f$ be
a slice hyperholomoprhic function on an open axially symmetric set that contains $\overline{U}$.
 Then
\begin{equation}\label{cauchynuovo}
 f(x)=\frac{1}{2 \pi}\int_{\partial (U\cap \mathbb{C}_\mathbf{j})} S_L^{-1}(s,x)\, ds_\mathbf{j}\,  f(s),\qquad\text{for any }\ \  x\in U.
\end{equation}
If $f$ is a right slice hyperholomoprhic function on a set that contains $\overline{U}$,
then
\begin{equation}\label{Cauchyright}
 f(x)=\frac{1}{2 \pi}\int_{\partial (U\cap \mathbb{C}_\mathbf{j})}  f(s)\, ds_\mathbf{j}\, S_R^{-1}(s,x),\qquad\text{for any }\ \  x\in U.
 \end{equation}
 Moreover, the integrals  depend neither on $U$ nor on the imaginary unit $\mathbf{j}\in\mathbb{S}$.
\end{theorem}

\subsection{Monogenic functions}\label{entireMON}

This is the second class of hyperholomorphic function  due to the  Fueter-Sce extension theorem, see \cite{ColSabStrupSce},
associated with the second map $T_{FS2}$.
 We recall that the generalized Cauchy-Riemann operator in $\mathbb R^{n+1}$ is defined by:
$$
\mathcal D:= \frac{\partial}{\partial{x_0}}+\sum_{i=1}^n e_i \frac{\partial}{\partial{x_i}}.
$$
\begin{definition}[Monogenic Functions, see \cite{BDS82}] Let $U\subseteq\mathbb R^{n+1}$ be an open subset.
 A function $f:\, U\to \mathbb{R}_n$, of class $C^1$, is called left monogenic if
$$
\mathcal D f= \frac{\partial}{\partial{x_0}} f +\sum_{i=1}^n e_i \frac{\partial}{\partial{x_i}} f=0.
$$
A function $f:\, U\to \mathbb{R}_n$, of class $C^1$, is called right monogenic if
$$
f \mathcal D = \frac{\partial}{\partial{x_0}} f +\sum_{i=1}^n \frac{\partial}{\partial{x_i}} f e_i=0.
$$
The set of  left monogenic functions (resp. right monogenic functions) will be denoted by $\mathcal M_L(U)$
(resp. $\mathcal M_R(U)$).
\end{definition}
\begin{definition}[Fueter's homogeneous polynomials]\label{fp}
Let $x\in\mathbb{R}^{n+1}$.
Given a multi--index $k=(k_1,...,k_n)$ where $k_i$ are integers, we set $|k|=\sum_{i=1}^nk_i$ and $k!=\prod_{i=1}^nk_i!$. We define the homogeneous polynomials $P_k(x)$ as follows:
\begin{enumerate}
\item[(I)] For a multi--index $k$ with at least one negative component we set
$$
P_k(x):=0.
$$
For $k=0=(0,...,0)$ we set
$$
P_0(x):=1.
$$
\item[(II)]
For a multi-index $k$ with $|k|>0$ and the integers $k_j$ nonnegative, we define $P_k(x)$ as follows: for each $k$ consider the sequence of indices $j_1,j_2,\ldots ,j_{|k|}$ be given such that $1$ appears in the sequence $k_1$ times, $2$ appears $k_2$ times, et cetera and, finally, $n$ appears $k_n$ times (i.e $\{j_1,j_2,\ldots ,j_{|k|}\}=\{ \underset{\textrm{$k_1$ times}}{\underbrace{1,\dots, \,1 }},\dots, \underset{\textrm{$k_n$ times}}{\underbrace{n,\dots, \,n }}\}$).
We define $z_i=x_i-x_0e_i$ for any $i=1,\dots , n$ and $z=(z_1,\dots, z_n)$. We set
$$
z^k:=z_{j_1}z_{j_2}\ldots z_{j_{|k|}}
$$
and
$$
|z|^k=|z_1|^{k_1}\cdots |z_n|^{k_n}
$$
these products contains $z_1$ exactly $k_1$-times, $z_2$ exactly $k_2$-times and so on. We define
$$
P_k(x):=\frac{1}{|k|!}\sum_{\sigma\in perm(k)} z_{j_{\sigma(1)}}z_{j_{\sigma(2)}} \ldots z_{j_{\sigma(|k|)}},
$$
where $perm(k)$ is the permutation group with $|k|$ elements.
\end{enumerate}
\end{definition}
\begin{remark}
In \cite{BDS82}  pag. 68, the definition of Fueter polynomials is given in the following way: let $s\in\mathbb N$, for any combination $(\ell_1,\dots,\ell_s)$ of $s$ elements out of (1,\dots, n) repetitions being allowed, we put
$$ V_0(x)=e_0 $$
and
$$ V_{\ell_1,\dots, \ell_s}(x)=\frac 1{s!} \sum_{\pi(\ell_1,\dots, \ell_s)} z_{\ell_1}\cdot\dots\cdot z_{\ell_s} $$
where the sum runs over all distinguishable permutation of all of $(\ell_1,\dots, \ell_s)$. According to this definition we have the following relation
$$ P_m(x):= m! V_{\ell_1,\dots, \ell_{s}}(x),$$
where $m=(m_1,\dots, m_n)$ and $m_i$ is the number of times that $i$ appear in the sequence $(\ell_1,\dots,\ell_{s})$ being $|m|=s$.
\end{remark}
The polynomials $P_k(x)$ play an important role in the monogenic function theory and we collect some of their properties in the next proposition (see Theorem $6.2$ in \cite{GHS08}):
\begin{theorem}\label{fpt}
Consider the Fueter polynomials $P_k(x)$ defined above. Then the following facts hold:.

(I) the recursion formula
$$
kP_k(x)=\sum_{i=1}^nk_i P_{k-\varepsilon_i}(x)z_i=\sum_{i=1}^nk_i z_iP_{k-\varepsilon_i}(x),
$$
and also
$$
\sum_{i=1}^nk_i P_{k-\varepsilon_i}(x)e_i=\sum_{i=1}^nk_i e_iP_{k-\varepsilon_i}(x),
$$
where $\varepsilon_i=(0,...,0,1,0,...,0)$ with $1$ in the position $i$.

(II) The derivatives $\partial_{x_j}$ for $j=1,...,n$, are given by
$$
\partial_{x_j}P_k(x)=k_jP_{k-\varepsilon_j}(x).
$$

(III) The Fueter polynomials $P_k(x)$ are both left and right monogenic.

(IV) The following estimates holds
$$
|P_k(x)|\leq |x|^{|k|}.
$$

(V)
For all paravectors $x$ and $y$, and for any multi--index $k$ there holds the binomial formula
$$
P_k(x+y)=\sum_{i+j=k}\frac{k!}{i!j!}P_i(x)P_j(y).
$$
\end{theorem}
\begin{remark}\label{remark_fp}
In particular, by the property (II) stated in the previous theorem we have for any multi-index $m, k\in\mathbb N_0^n$ that
$$
\partial^m_{\underline x} P_{k}(x)= \frac{k!}{(k-m)!}P_{k-m}(x).
$$
 We also remark that $P_k(x)$ is a paravector for any $x\in\mathbb R^{n+1}$ (see Proposition 10.6 in \cite{GHS08}). In particular we have
$$
P_k(\unx)=x_1^{k_1}\cdots x_n^{k_n}
$$
and
$$
P_k(x_0)=x_0^{|k|} \frac{1}{|k|!} \sum_{\sigma \in\operatorname{Perm} (k)} e_{\ell_{\sigma(1)}}\cdots e_{\ell_{\sigma(|k|)}}$$ where $(\ell_1, \dots, \ell_{|k|})
$
is such that
$$
\ell_{\sum_{r=1}^{j-i}k_{r}+1}= \dots= \ell_{\sum_{r=1}^{j} k_{r}}=j
$$
for any $j=1, \dots, n$. We define
\begin{equation}\label{sum_ima_units_2}
e_k:=\sum_{\sigma\in \operatorname{Perm}(|k|)} e_{\ell_{\sigma(1)}} \cdots e_{\ell_{\sigma(|k|)}}
\end{equation}
(note that here $k=(k_1,\dots, k_n)$ is a multi-index so $e_k$ has not to be confused with the imaginary units $e_j$ for $j=1,\dots ,n$).
\end{remark}

Another important tool for what follows in Section \ref{smonogenic} is the so called left Cauchy-Kowalewski extension also called CK-extension (see pag. 111 in \cite{BDS82}).
\begin{theorem}[CK-Extension]\label{ck_extension}
 Given an analytic function $f$ in $\mathbb R^n$,
 i.e
 $$f(\unx)=\sum_{m\in\mathbb N^n} x_1^{m_1}\cdots x_n^{m_n} \lambda_m
 $$
 where $\lambda_m\in\mathbb R_n$ for any $\unx\in\mathbb R^n$,
  there exists a unique left monogenic function $f^*$ which extends $f$ to $\mathbb R^{n+1}$ given by
$$ f^*(x)=\sum_{k=0}^{\infty} (-1)^k \overline D\left[ \frac{x_0^{2k+1}}{(2k+1)!} \Delta^k_n (f(\unx))\right]. $$
\end{theorem}
Using the CK-extension we can define the so called CK-product which is a product that preserve the left monogenic regularity.
\begin{definition}[CK-product]
 Given $f,g\in\mathcal M_L(\mathbb R^{n+1})$ the CK-product is defined as
$$f\odot_Lg =(f(\unx) \cdot g(\unx))^*.$$
\end{definition}
For the properties of this product we refer the reader to \cite{BDS82}. Here we recall the following two facts:
\begin{equation}\label{Fueter_poly_ck_prod}
P_k(x)=z_{\ell_1}\odot_L\dots \odot_L z_{\ell_{|k|}}
\end{equation}
and, if $f(\unx)\cdot g(\unx)=g(\unx)\cdot f(\unx)$, then
$$ f(x)\odot_L g(x)=g(x)\odot_L f(x). $$
Thus in particular $z_i\odot_L z_k=z_k\odot_L z_i$ for any $i,\, k=1,\dots, n$.

\subsection{Entire slice hyperholomorphic functions and entire monogenic functions}\label{entireexpMON}

Entire slice hyperholomorphic functions and entire monogenic functions will play an important role in our considerations.
\begin{definition}
Let $f$ be an entire  slice left monogenic function or an entire left monogenic function.
We say that $f$ is of finite order if there exists $\kappa>0$ such that
$$
M_f(r)< e^{r^\kappa}
$$
for sufficiently large $r$. The greatest lower bound $\rho$ of such numbers $\kappa$ is called order of $f$.
Equivalently, we can define the order as
$$
\rho=\limsup_{r\to\infty}\frac{\ln\ln M_f(r)}{\ln r}.
$$
\end{definition}
\begin{lemma}\label{growth_taylor_series}
The series
$$
\sum_{m\in\mathbb N^n_0} \frac{c(n,m)}{m!} P_m(x)
$$
is converging for any $x\in\mathbb R^{n+1}$ and it is a monogenic function that we call $g$ with order $1$ where we assume that
$$
c(n,m)=\frac {(n+|m|-1)!}{(n-1)! m!}.
$$
\end{lemma}
\begin{proof}
To prove the first assertion it is sufficient to observe that
$$
\limsup_{|m|\to+\infty} \left( \frac{c(n,m)}{m!} \right)^{\frac 1{|m|}}=0.
$$
Indeed we have that for any multi-index $m\in\mathbb N^n_0$ with $|m|=k$
 \begin{equation}\label{binom_est}
 n^k=\sum_{|m|=k} \frac{k!}{m!} \geq \frac {k!}{m!}
 \end{equation}
 which implies $m!\geq \frac {|m|!}{n^{|m|}}$. Thus, by the Stirling's formula, we have
\[
\begin{split}
& \limsup_{|m|\to+\infty} \left( \frac{c(n,m)}{m!} \right)^{\frac 1{|m|}} = \limsup_{|m|\to+\infty} \left( \frac{(n+|m|-1)!}{(n-1)!(m!)^2} \right)^{\frac 1{|m|}}\leq \limsup_{|m|\to+\infty} \left( \frac{(n+|m|-1)!}{(n-1)!(|m|!)^2} \right)^{\frac 1{|m|}}\\
&=\limsup_{|m|\to+\infty} \frac{n^2 (n+|m|-1)^{\frac{n+|m|-1}{|m|}}}{n^{\frac {n}{|m|}} |m|^2}= \limsup_{|m|\to+\infty} |m|^{-1+\frac{n-1}{|m|}} \frac{n^2 (1+\frac{n-1}{|m|})^{\frac{n+|m|-1}{|m|}}}{n^{\frac {n}{|m|}}} =0.
\end{split}
\]
By Theorem 1 in \cite{CAK07}, to prove the order of $g(x)=\sum_{m\in\mathbb N^n_0} P_m(x)a_m$ where $a_m=\frac{c(n,m)}{m!}$ is equal to $1$ it is sufficient to prove that
$$
\limsup_{|m|\to+\infty} \frac{|m|\log(|m|)}{-\log(\frac{|a_m|}{c(n,m)})}=0.
$$
By \eqref{binom_est} and Stirling's Formula we have
\[
\begin{split}
&\limsup_{|m|\to+\infty} \frac{|m|\log(|m|)}{-\log(\frac{a_m}{c(n,m)})}=\limsup_{|m|\to+\infty} \frac{|m|\log(|m|)}{-\log\left (\frac{c(n,m)}{m! c(n,m)}\right )}\leq \limsup_{|m|\to+\infty} \frac{|m| \log(|m|)}{\log(|m|!)-|m|\log(n)}\\
&=\limsup_{|m|\to+\infty} \frac{|m| \log(|m|)}{|m|\log(|m|)-|m|\log(n)}=1.
\end{split}
\]
\end{proof}
We recall that $A_1^L(\mathbb R^{n+1})$, the space of entire left slice hyperholomorphic functions of order $1$, or simply $A_1^L$ is the inductive limit of the following family of Banach spaces parametrized by $\rr\ni\sigma>0$:
$$A_{1,\sigma}^L(\mathbb R^{n+1})=\{ f\in \mathcal{SM}_L(\mathbb R^{n+1}):\, \|f\|_{1,\sigma}:=\sup_{x\in\mathbb R^{n+1}} |f(x)| e^{-\sigma |x|}<+\infty\}.$$
We denote by the same symbol $A_1^L$ the space of entire left monogenic functions of order $1$, that is the inductive limit of the following family of Banach spaces:
$$A_{1,\sigma}^L(\mathbb R^{n+1})=\{ f\in \mathcal{M}_L(\mathbb R^{n+1}):\, \|f\|_{1,\sigma}:=\sup_{x\in\mathbb R^{n+1}} |f(x)| e^{-\sigma |x|}<+\infty\}.$$
According to the context, $A_{1,\sigma}$ or $A_1$ will denote the space of entire monogenic functions or the space of  entire slice hyperholomorphic functions.

In this section they will denote both the regularity. It is also possible to define the spaces $A^R_{1,\sigma}(\mathbb R^{n+1})$ and $A_{1}^R$ by requiring that the function $f$ belongs to $\mathcal{SM}_R(\mathbb R^{n+1})$ or to $\mathcal M_R(\mathbb R^{n+1})$. Along all the paper, we denote by $A_1$ and $A_{1,\sigma}$ the spaces $A_1^L$ and $A_{1,\sigma}^L$.
We say that a sequence $\{f_N\}_{N\in\mathbb N}\subseteq A_1$ converges to $f$ in $A_1$ if there exists $\sigma>0$ such that $f\in A_{1,\sigma}$ and $\{f_N\}_{N\in\mathbb N}\subseteq A_{1,\sigma}$, moreover, $$\lim_{N\to \infty} \| f_N(x)-f(x)\|_{1,\sigma}=0.$$
 We denote by $H(\mathbb R^{n+1})$ both the set of left monogenic functions or slice hyperholomorphic functions equipped by the uniform convergence in the compact subset of $\mathbb R^{n+1}$.  We need the following proposition and lemma which are the Cliffordian version of Proposition 2.1 and Lemma 2.4 in \cite{CSSY22}.
\begin{proposition}\label{conv_A_1}
Let $\{f_N\}_{N\in \mathbb N}$ be a sequence of elements in $A_1(\mathbb R^{n+1})$. The two following assertions are equivalent:
\begin{itemize}
\item the sequence converges towards $0$ in $A_1(\mathbb R^{n+1})$;
\item the sequence converges towards $0$ in $H(\mathbb R^{n+1})$ and there exists $A_f\geq 0$ and $B_f\geq 0$ such that
$$ \forall\, N\in \mathbb N,\quad \forall\, x\in\mathbb R^{n+1} ,\,\, |f_N(x)|\leq A_fe^{B_f|x|}. $$
\end{itemize}
\end{proposition}
\begin{proof}
We assume that the first assertion holds true. This means that there exists $\sigma_1>0$ and $N_1>0$ such that for any $N>N_1$ we have $\|f_N\|_{1,\sigma_1} <1$. Let $\sigma:=\max  \{\alpha_0,\dots, \alpha_{N_1},\sigma_1\}$ where $\alpha_j>0$ is one of the positive constant such that $f_j\in A_{1,\alpha_j}$ for any $j=1,\dots N_1$, thus we have $\{f_N\}_{N\in\nn}\subseteq A_{1,\sigma}(\mathbb R^{n+1})$. In particular, we have: $|f_N|\leq A e^{\sigma |x|}$ for any $N\in\mathbb N_0$ where $A=\sup\{ \|f_0\|_{1,\sigma},\dots, \|f_{N_1}\|_{1,\sigma}, 1\}$. Now, let $K$ be a compact subset of $\mathbb R^{n+1}$, we choose an arbitrary small positive constant $\epsilon>0$ and we define $\delta=\epsilon e^{-\sigma\sup_{x\in K} |x|}$. We know that there exists a positive integer $N_\delta$ such that for any $N>N_\delta$ we have $\|f_N\|_{1,\sigma} < \delta$. Thus we have for any $N>N_\delta$
$$\sup_{x\in K} |f_N(x)|<\sup_{x\in K} |f_N(x)| e^{\sigma\sup_{x\in K} |x|-\sigma |x| }\leq e^{\sigma\sup_{x\in K} |x|}  \sup_{x\in \mathbb R^{n+1}} |f_N(x)| e^{-\sigma|x|}\leq e^{\sigma\sup_{x\in K} |x|} \delta=\epsilon. $$
This prove that the sequence $\{f_N\}_{N\in\nn}$ converges towards 0 in $H(\mathbb R^{n+1})$.

Now, we assume that the second assertion holds true. Let $\epsilon$ be a fixed positive constant and $B$ another positive constant such that $B\geq B_f$. Thus we know that there exists a third positive constant $R_{\epsilon, B}$ such that
$$ \forall\, N\in\mathbb N,\quad \sup_{|x|\geq R_{\epsilon,B}}|f_N(x)| e^{-B|x|}\leq A_f\sup_{|x|\geq R_{\epsilon,B}} e^{(B_f-B)|x|} \leq A_f e^{(B_f-B)R_{\epsilon, B}}<\epsilon. $$
On the other hand, since the ball centered in $0$ with radius $R_{\epsilon,B}$ is a compact subset of $\mathbb R^{n+1}$ we know that there exists $N_\epsilon$ such that for any $N>N_\epsilon$ we have
$$ \sup_{|x|\leq R_{\epsilon,B}}|f_N(x)| e^{-B|x|}<\sup_{|x|\leq R_{\epsilon,B}}|f_N(x)|<\epsilon. $$
Therefore $\sup_{N\geq N_\epsilon} \|f_N\|_{\sigma, B}<\epsilon$ and the sequence converges to 0 in $A_1(\mathbb R^{n+1})$.
\end{proof}
Here we state an important property for the Taylor coefficients of an entire slice hyperholomorphic or monogenic exponentially increasing function (see Lemma 2.13 and Lemma 4.14 in \cite{CPSS21}).
\begin{lemma}\label{taylor_coeff_estimate}
A function
$$ f(x)=\sum_{k=0}^\infty x^k a_k\quad (\textrm{resp. } f(x)=\sum_{m\in \mathbb N_0^n}^\infty P_m(x) a_m) $$
belongs to $A_{1}$ if and only if there exists positive constants $C_f$, $b_f$ such that
$$ |a_k|\leq C_f\frac{b^k_f}{k!} \quad \left (\textrm{resp. } |a_m|\leq C_f\frac{b_f^{|m|} c(n,m)}{|m|!}\right ),$$
where
$$
c(n,\ell)=\frac{n(n+1)\cdots (n+|\ell|-1)}{\ell!}=\frac{(n+|\ell|-1)!}{(n-1)!\, \ell ! }.
$$
Furthermore, a sequence $f_m$ in $A_1$ tends to zero if and only if $C_{f_m}\to 0$ and $b_{f_m}< b$ for some $b>0$.
\end{lemma}

\subsection{Some special functions in the monogenic setting}
 The monogenic exponential function can be defined by means of the C-K extension of the slice hyperholomorphic exponential function restricted to the purely imaginary paravectors:
$$
e^{\underline x}:= \sum_{k=0}^\infty \frac 1{k!} \left( \underline x \right)^k=\sum_{k=0}^\infty \frac 1{k!}\sum_{m\in\mathbb N^n, |m|=k}\sum_{\ell_1=1}^n \dots \sum_{\ell_k=1}^n x_{\ell_1}\dots x_{\ell_k} e_{\ell_1}\dots e_{\ell_k}.
$$
Since
$$
(z_1e_1+\dots z_ne_n)^{\odot_L k}\Big|_{x_0=0}=\sum_{\ell_1=1}^n\cdots \sum_{\ell_k=0}^n x_{\ell_1}\dots x_{\ell_k} e_{\ell_1}\dots e_{\ell_k}
$$ and
$(z_1e_1+\dots z_ne_n)^{\odot_L k}$ is a monogenic function, we have that, by the Identity Theorem of monogenic functions (see page 180 in \cite{GHS08}), the C-K extension of $e^{\unx}$ is given by
\[
\begin{split}
E(x)&=\sum_{k=0}^\infty \frac 1{k!} \left( z_1e_1+\dots+z_ne_n \right)^{\odot_L k}=\sum_{k=0}^\infty \frac 1{k!} \sum_{\ell_1=1}^n \dots \sum_{\ell_k=1}^n (z_{\ell_1}e_{\ell_1})\odot_L\dots\odot_L(z_{\ell_k}e_{\ell_k})\\
&=\sum_{k=0}^\infty \frac 1{k!} \sum_{m\in\mathbb N^n_0, |m|=k} P_m(x) \sum_{\sigma\in \pi '(1,\dots, k)} e_{\ell_{\sigma(1)}} \cdots e_{\ell_{\sigma(k)}},
\end{split}
\]
where the second inequality is due to the distributivity and associativity property of the CK product, the third equality is due to the linearity on the right for the C-K product and  the equation \eqref{Fueter_poly_ck_prod}, moreover, the last sum runs over all distinguishable permutations of all of $(\ell_1, \dots, \ell_{k})$ and $$\ell_{\sum_{r=1}^{i-i}m_{r}+1}= \dots= \ell_{\sum_{r=1}^{i} m_{r}}=i$$
for any $i=1, \dots, n$. We can also say that $m_i$ is the number of $i$ that appear in the ordered set: $(\ell_1,\dots, \ell_k)$. By the multinomial Newton formula, the number of addends in the sum
\begin{equation}\label{sum_ima_units}
e_m':=\sum_{\sigma\in \pi'(1,\dots, k)} e_{\ell_{\sigma(1)}} \cdots e_{\ell_{\sigma(k)}}
\end{equation}
is $\frac{k!}{m_1!\cdots m_n!}$.
\begin{proposition}\label{exp_mon_funct}
The radius of convergence of the series that defines $E(x)$ is $+\infty$ and its order is $1$.
\end{proposition}
\begin{proof}
By \eqref{binom_est} and the estimate
\[
\sum_{m\in\mathbb N^n_0, \, |m|=k} 1=\binom{n+k-1}{k}\leq (k+1)^{n-1}
\]
(see for example \cite{AIO20} pag. 12), we have in particular that:
\[
\left| \sum_{k=0}^\infty \frac 1{k!} \sum_{m\in\mathbb N^n_0, |m|=k} P_m(x) e'_m \right| \leq \sum_{k=0}^\infty \left( \frac {n^k}{k!} \sum_{|m|=k} 1 \right) |x|^{k} \leq \sum_{k=0}^\infty \left( \frac {n^k (k+1)^{n-1}}{k!}\right) |x|^{k}
\]
and
\[\limsup_{k\to+\infty} \frac {n^k (k+1)^{n-1}}{k!}
 = \limsup_{k\to+\infty} \frac{n\left( (k+1)^{\frac 1k} \right)^{n-1}}{k!^{\frac 1k}}=0,
\]
where we used $\lim_{k\to +\infty} (k!)^{\frac 1k}=+\infty$ and $\lim_{k\to+\infty}(1+k)^{\frac 1k}=1$. This implies that the radius of convergence of the series that define $E(x)$ is $+\infty$.

 Now, we prove that the order of $E(x)$ is 1. Indeed, by Theorem 1 in \cite{CAK07} it is sufficient to prove that
$$
\limsup_{|m|\to+\infty} \frac{|m|\log(|m|)}{-\log(\frac{|a_m|}{c(n,m)})}=1
$$
where
$
a_m:=\frac {e'_m}{|m|!}$ and
$
\frac{|m|\log(|m|)}{-\log(\frac{|a_m|}{c(n,m)})}:=0$
if $|a_m|=0.
$
 We observe that
 $$
 |\frac{a_m}{c(n,m)}|=\left |\frac{m! e'_m }{(n+|m|-1)!|m|!}\right |\leq \frac{1}{(n+|m|-1)!}.
 $$
 Thus we have, when $a_m\neq 0$
$$
\limsup_{|m|\to+\infty} \frac{|m|\log(|m|)}{-\log(\frac{|a_m|}{c(n,m)})}=\limsup_{|m|\to+\infty} \frac{|m|\log(|m|)}{\log(\frac{c(n,m)} {|a_m|})}\leq \limsup_{|m|\to+\infty} \frac{|m|\log(|m|)}{\log((n+|m|-1)!)}=1.
$$
On the other hand we have that for the multi-indexes of the type $m=(m_1,1,0,\dots,0)$ with $m_1+1$ be an odd number, the number of even permutations with repetitions of $m$ is one bigger then the number of odd permutations with repetitions of $m$. Thus, for those multi-indexes  $m$, we have
$$|e'_m|=\left| \sum_{\sigma\in\pi'(|m|)} e_{\ell_{\sigma(1)}} \cdots e_{\ell_{\sigma(|m|)}} \right|=\left|e_1^{m_1}\cdots e_n^{m_n}\left( \sum_{\sigma\in\pi'(|m|), \,\operatorname{even}}1- \sum_{\sigma\in\pi'(|m|), \,\operatorname{odd}}1 \right) \right| =1.
$$
We can conclude that
\begin{eqnarray*}
& \limsup_{|m|\to+\infty} \frac{|m|\log(|m|)}{-\log(\frac{|a_m|}{c(n,m)})}
\\
&
\geq \limsup_{|m|\to+\infty,\, m=(m_1,1,0,\dots,0),\, \operatorname{odd}} \frac{|m|\log(|m|)}{-\log(\frac{|a_m|}{c(n,m)})}\\
& =\limsup_{|m|\to+\infty,\, m=(m_1,1,0,\dots,0),\, \operatorname{odd}} \frac{|m|\log(|m|)}{\log(c(n,m)|m|!)}
\end{eqnarray*}
\begin{eqnarray*}
& = \limsup_{|m|\to+\infty,\, m=(m_1,1,0,\dots,0),\, \operatorname{odd}} \frac{|m|\log(|m|)}{\log \left ( (n+|m|-1)! \frac{|m|!}{m!} \right)}\\
& = \limsup_{|m|\to+\infty,\, m=(m_1,1,0,\dots,0),\, \operatorname{odd}} \frac{|m|\log(|m|)}{\log \left ( (n+|m|-1)!\right) + \log\left( \frac{|m|!}{m!} \right)}\\
&\overset{\eqref{binom_est}}{\geq} \limsup_{|m|\to+\infty,\, m=(m_1,1,0,\dots,0),\, \operatorname{odd}} \frac{|m|\log(|m|)}{\log \left ( (n+|m|-1)!\right) +|m|\log(n)}  =1.
\end{eqnarray*}
\end{proof}
\begin{proposition}\label{derivative_mon_exp}
For any constant $\alpha\in\mathbb R$ we have
\[
\frac 1{c_k} \sum_{|\mathbf{m}|=k} \frac 1{\mathbf{m}!} \partial^{\mathbf{m}}_{\underline x} \left( \sum_{\ell_1=1}^n \dots \sum_{\ell_k=1}^n (\alpha z_{\ell_1}e_{\ell_1}) \odot_L\dots\odot_L(\alpha z_{\ell_k}e_{\ell_k}) \right)\Bigg|_{x=0} = \alpha^k
\]
for
$$c_k:=\begin{cases} (-n)^q\quad k=2q\\
  (-n)^q(e_1+\dots +e_n)\quad\textrm{if $k=2q+1$}
  \end{cases}$$ and
  \[
  \sum_{|\mathbf{m}|=s} \frac 1{\mathbf{m}!} \partial^{\mathbf{m}}_{\underline x} \left( \sum_{\ell_1=1}^n \dots \sum_{\ell_k=1}^n (\alpha z_{\ell_1}e_{\ell_1}) \odot_L\dots\odot_L(\alpha z_{\ell_k}e_{\ell_k}) \right)\Bigg|_{x=0} = 0
  \]
 for $s\neq k$. In particular, we obtain
 $$
 \frac 1{c_k} \sum_{|\mathbf{m}|=k} \frac 1{\mathbf{m}!} \partial^{\mathbf{m}}_{\underline x} \left( E(\alpha x) \right)\Big|_{x=0}=\alpha^k.
 $$
\end{proposition}
\begin{proof}
By the definition of the left C-K product and the C-K extension we have that
 \begin{equation}\label{ck_prod_fueter_var}
 (\alpha z_{\ell_1}e_{\ell_1})\odot_L\dots\odot_L(\alpha z_{\ell_k}e_{\ell_k}):=\alpha^k \sum_{s=0}^{\infty} (-1)^s \overline D\left[ \frac{x_0^{2s+1}}{(2s+1)!} \Delta^s_n(x_{\ell_1} \cdots x_{\ell_k} \cdot e_{\ell_1}\cdots e_{\ell_k}) \right].
\end{equation}
 Thus if $\mathbf{m}=(m_1,\dots, m_n)\in\mathbb N^n$ and $m_i$ is the number of $\ell_j$ such that $\ell_j=i$ for $j=1,\dots, k$ then, differentiating by $\partial^{\mathbf{m}}_{\unx}$ both side of the equation \eqref{ck_prod_fueter_var} and evaluating it at $0$, we have
 $$
 \partial^{\mathbf{m}}_{\underline x}\left(\alpha  z_{\ell_1}e_{\ell_1})\odot_L\dots\odot_L (\alpha z_{\ell_k}e_{\ell_k}) \right)\Bigg|_{x=0}=\alpha^k \mathbf{m}! e_{\ell_1}\cdots e_{\ell_k}
 $$
 and
 $$ \partial^{\mathbf{m}}_{\underline x}\left( \sum_{\sigma \in\operatorname{Perm}(k)}\alpha z_{\ell_{\sigma(1)}}e_{\ell_{\sigma(1)}}\odot_L\dots\odot_L \alpha z_{\ell_{\sigma(k)}}e_{\ell_{\sigma(k)}} \right)\Bigg|_{x=0}=\alpha^k \mathbf{m}! \sum_{\sigma \in\operatorname{Perm}(k)} e_{\ell_{\sigma(1)}}\cdots e_{\ell_{\sigma(k)}}.
 $$
 In conclusion we get
 \[
 \begin{split}
 & \sum_{|\mathbf{m}|=k} \frac 1{\mathbf{m}!} \partial^{\mathbf{m}}_{\underline x} \left( \sum_{\ell_1=1}^n \dots  \sum_{\ell_k=1}^n (\alpha z_{\ell_1}e_{\ell_1}) \odot_L\dots\odot_L(\alpha z_{\ell_k}e_{\ell_k}) \right)\Bigg|_{x=0} =\alpha^k \sum_{\ell_1=1}^n \dots \sum_{\ell_k=1}^n e_{\ell_1}\cdots e_{\ell_k}\\
 &=\alpha^k\cdot \begin{cases} (-n)^q\quad k=2q\\
  (-n)^q(e_1+\dots +e_n)\quad\textrm{if $k=2q+1$},
  \end{cases}
\end{split}
\]
where the last equality can be found in \cite{BDS82} page 117. Moreover, if $k>s$ the number of derivatives exceed the degree of the polynomials thus the differentiation is zero and if $k<s$ when we evaluate at $0$ we obtain zero again. Eventually, by the previous formulas we have
 $$
 \frac 1{c_k} \sum_{|\mathbf{m}|=k} \frac 1{\mathbf{m}!} \partial^{\mathbf{m}}_{\underline x} \left( E(\alpha x) \right)\Big|_{x=0}=\frac{1}{c_k}\sum_{|m|=k} \sum_{s=0}^\infty \frac 1{m!} \partial^m_{\unx} \left( \sum_{\ell_1=0}^n\dots \sum_{\ell_s=0}^n (\alpha z_{\ell_1})\odot_L\dots\odot_L (\alpha z_{\ell_s})  \right)=\alpha^k.
 $$
\end{proof}
We recall, as we observed in Section \ref{sslicemonogenic}, that
$$
e^{\underline x}= \sum_{k=0}^\infty \frac 1{(2k)!} (\underline x)^{2k}+\sum_{k=0}^\infty \frac 1{(2k+1)!} (\underline x)^{2k+1}=\cosh(\unx)+\sinh(\unx),
$$
thus, reasoning as for the definition of the function $E(x)$, we have that the C-K extensions of $\cosh(\unx)$ and $\sinh(\unx)$ are given by
$$
\operatorname{Cosh}(x):=\sum_{k=0}^\infty \frac 1{(2k)!} (z_1e_1+\dots+z_ne_n)^{\odot_L 2k}
$$
and
$$
 \operatorname{Sinh}(x):= \sum_{k=0}^\infty \frac 1{(2k+1)!} (z_1e_1+\dots+z_ne_n)^{\odot_L 2k+1}.
$$
We now need some estimates for the CK-extensions.
\begin{proposition}
The following estimate holds
\begin{equation}
\left| (z_1e_1+\dots+z_ne_n )^{\odot_L k} \right| \leq (|z_1|+\dots+|z_n|)^k. \label{hyper_e1}
\end{equation}
Moreover, let $f,g\in \mathcal SH_L(\mathbb R^{n+1})$ we consider $F(x)$ and $G(x)$ the CK-extensions of the restrictions:
$$
f(\unx)=f(x)\Big |_{x_0=0}\ \ {\rm and}\ \  g(\unx)=g(x)\Big|_{x_0=0}.
$$
 Then we have
\begin{equation}
\label{hyper_e2}
\left| F(x)\odot_L G(x) \right| \leq 2^{\frac n2} \tilde F(x) \tilde G(x).
\end{equation}
where
$$
\tilde F(x):=\sum_{j=0}^\infty (|z_1|+\dots + |z_n|)^j |a_j|\ \ {\rm and}\ \
\tilde G(x):=\sum_{j=0}^\infty (|z_1|+\dots + |z_n|)^j |b_j|.
$$
\end{proposition}
\begin{proof}
First we prove \eqref{hyper_e1}. We have that
\[
\begin{split}
\left| (z_1e_1+\dots+z_ne_n)^{\odot_L k}\right| \leq \sum_{\ell_1=1}^n\dots \sum_{\ell_{k}=1}^n |z_{\ell_1} \odot_L \cdots  \odot_L z_{\ell_{k}} (e_{\ell_1}  \cdots  e_{\ell_{k}})| & \leq \sum_{\ell_1=1}^n\dots \sum_{\ell_{k}=1}^n |z_{\ell_1}|\cdots |z_{\ell_{k}}|\\
& =(|z_1|+\dots+|z_n|)^k.
\end{split}
\]
We now observe that if $f(x)=\sum_{j=0}^{\infty} x^ja_j$ and $g(x)=\sum_{j=0}^{\infty} x^jb_j$ then, by the uniqueness of the CK extension we obtain
 $$
 F(x)=\sum_{j=0}^\infty (z_1e_1+\dots +z_ne_n)^{\odot_L j} a_j\ \ {\rm and}\ \ G(x)=\sum_{j=0}^\infty (z_1e_1+\dots +z_ne_n)^{\odot_L j} b_j.
 $$
 Thus, we get the estimates
\[
\begin{split}
\left| F(x) \odot_L G(x) \right| &=\left | \sum_{j=0}^\infty \sum_{k=0}^\infty (z_1e_1+\dots + z_ne_n)^{\odot_L j+k} a_j b_k \right|\\
&\leq  \sum_{j=0}^\infty \sum_{k=0}^\infty \left | (z_1e_1+\dots + z_ne_n)^{\odot_L j+k} a_j b_k \right|\\
&\overset{\eqref{hyper_e1}}{ \leq} 2^{\frac n2} \sum_{j=0}^\infty \sum_{k=0}^\infty (|z_1|+\dots +|z_n|)^{i+j} |a_j| |b_k|\\
&= 2^{\frac n2} \tilde F(x) \tilde G(x).
\end{split}
\]
\end{proof}
\begin{remark}\label{hyper_e3}
The previous proposition can be applied when the functions $f$ and $g$ are $\sinh(x)$, $\cosh(x)$ or $e^x$. These are particular functions with positive real Taylor coefficients thus in this case we get
$$
\left| F(x)\odot_L G(x) \right| \leq \tilde F(x) \tilde G(x).
$$
\end{remark}
By the properties of the left slice hyperholomorphic exponential function:
$$
e^{\underline x}e^{\underline y}=e^{\underline x+\underline y}=e^{\underline y}e^{\underline x}\quad\textrm{end}\quad \cosh(\underline x)=\frac{e^{\underline x}+e^{-\underline x}}{2},\quad \sinh(\underline x)=\frac{e^{\underline x}-e^{-\underline x}}{2}
$$
if $\underline x,\underline y\in\mathbb R^{n}$ commute, thus we have by the first equalities $$
E(ax)\odot_L E(bx)=E((a+b)x)=E(bx)\odot_L E(ax)
$$
 for any $a,\, b\in\mathbb R$. Another Lemma (see Lemma 3.13 in \cite{CKPS24}) which will be useful in Section \ref{smonogenic} is the following one:
\begin{lemma}\label{l12}
Let
$$
g_1( x)=\sum_{| m|=0}^{\infty} P_m (x)a_{m}\in A_{1, \tau_1}\ \ {\rm and}\ \  g_2(x)=\sum_{|m|=0}^{\infty} P_{m}(x)b_{m}\in A_{1, \tau_2}.
$$
 Let $\delta$ be a positive constant then for any $\eta>0$ there exists $C(n, \eta,\tau_1,\tau_2)>0$ such that
$$ \|g_1\odot_{L} g_2 \|_{A_{\rho, (1+\eta)(n+\delta+1) (\tau_1+\tau_2)}}\leq C(n, \eta,\tau_1,\tau_2) \|g_1\|_{\rho,\tau_1} \|g_2\|_{\rho,\tau_2}. $$
\end{lemma}

\section{Superoscillations in the slice hyperholomorphic setting}\label{sslicemonogenic}

\subsection{One Clifford variable in the slice hyperholomorphic setting}
In this section we investigate superoscillations in the slice hyperholomorphic setting for functions with values in the Clifford algebras.
 We define, for all $\unx\in\mathbb R^{n}$ with $\unx\neq 0$, $\omega(\unx):=\frac \unx{|\unx|}$. We consider the function
$$ F_N(x_1,\dots, x_n,a):= \left( \cos \left(\frac{|\unx|}{N}\right) +\omega(\unx) a\sin \left(\frac{|\unx|}{N}\right) \right)^N, $$
where $a$ is a real constant. Note that when $\unx\to 0$ we have: $|\omega(\unx)|a\sin\left(\frac{\unx}{N}\right)  \to 0$, thus to preserve the continuity and the differentiability we define $F_N(\underline 0,a):=1$. Recalling that
$$
\cos\left(\frac{|\unx|}{N}\right)=\frac{e^{\omega(\unx)|\unx|/N}+e^{-\omega(\unx)|\unx|/N }}{2} \quad{\rm and}\quad \omega(\unx)\sin \left(\frac{|\unx|}{N}\right) =\frac{e^{\omega(\unx) |\unx|/N }-e^{-\omega(\unx) |\unx|/N }}{2} $$
and developing the Nth-power of the binomial in the sum of $F_N$, we have
\begin{equation}\label{rep_superos}
F_N(x_1,\dots ,x_n,a)=\sum_{j=0}^N  C_j(N,a) e^{(1-2j/N) \unx}
\end{equation}
where
\begin{equation}\label{binom_coeff}
C_j(N,a)=\binom{N}{j} \left(\frac{1+a}2\right)^{N-j} \left(\frac{1-a}2\right)^{j}.
\end{equation}
If we fix $\unx\in\mathbb R^{n}$ and we let go $N$ to infinity we obatin
$$\lim_{N\to \infty} F_N(x_1,\dots ,x_n,a)=e^{a\unx},$$
this can be achieved as for the complex case because when $N\to\infty$ then $\frac{\unx}{N}\to 0$ in the fixed complex slice $\mathbb C_{\omega(\unx)}$.
This simple observation is significant because we now have a sequence of Clifford valued functions whose frequencies $(1-2j/N)$ are bounded by one, and yet their limit has the arbitrarily large frequency $a$. This explains why the sequence $\{F_N\}_{N\geq 1}$ is called superoscillatory looking at representation \eqref{rep_superos}. We observe that the functions $F_N(x_1,\dots ,x_n,a)$ admit a simple extension to an entire slice hyperholomorphic function (note that the coefficients $C_j(N,a)$ are real)  that we denote by the same symbol
\begin{equation}\label{super_os_2}
F_N(x_0,x_1,\dots,x_n,a):=\sum_{j=0}^N C_j(N,a) e^{(1-2j/N) x}=\left(\cosh\left(\frac xN\right)+a\sinh \left(\frac xN\right)\right)^N,
\end{equation}
 where $e^x:=\sum_{k=0}^{+\infty} \frac{x^k}{k!}$, $\cosh(x)=\sum_{k=0}^{\infty} \frac{x^{2k}}{(2k)!}$ and $\sinh(x)=\sum_{k=0}^{\infty} \frac{x^{2k+1}}{(2k+1)!}$. We  prove that the sequence $\{F_N\}_{N\in\mathbb N^*}$ is converging in $A_1$ to $e^{ax}$.
\begin{lemma}\label{l1}
Let $a\in\mathbb R$ and $\alpha:=\max(1,|a|)$. Then the sequence $\{F_N\}_{N\in\mathbb N^*}$, defined in \eqref{super_os_2}, is converging in $A_1$ to $e^{ax}$. In particular the following estimates hold
\begin{equation}\label{e1}
\begin{split}
& |F_N(x,a)|\leq \exp(|a||x|+ |x_0|)\leq \exp((|a|+1)|x|)\\
& |F_N(x,a)-e^{ax}|\leq \frac 23 \frac {|a^2-1|}{N}|x|^2 \exp((\alpha+1)|x|).
\end{split}
\end{equation}
\end{lemma}
\begin{proof}
By Proposition \ref{conv_A_1}, to prove the first part of the Lemma it is sufficient to prove the two inequalities in \eqref{e1}. We prove the inequalities in \eqref{e1} along each slice. Thus, we fix $J\in \mathbb S$ and we consider the slice $\cc_J$. If $z\in \cc_J$ then we can write $F_N(z,a)=\left( \cosh(\frac zN)+a\sinh(\frac zN) \right)^N$. We observe that $\operatorname{sinhc}(z):=\frac {\sinh(z)}{z}=\int_0^1 \cosh(tz)\, dt$ and, by the Euler formula, we have the estimate $|\operatorname{sinhc}(z)|\leq e^{|\operatorname{Re}(z)|}$. Thus, for any $N\in\mathbb N$ and for any $z\in \mathbb C_J$ we have
\begin{equation}\label{e2}
\begin{split}
 |F_N(z,a)|&=\left| \cosh\left(\frac zN \right)+a\sinh\left(\frac zN\right) \right|^N\\
& =\left| \cosh\left(\frac zN \right)+\frac {z}N a\operatorname{sinhc}\left(\frac zN\right) \right|^N\\
& \leq e^{|\operatorname{Re}(z)|}\left( 1+\frac {|az|}{N} \right)^N\leq \exp(|a| |z|+|\operatorname{Re}(z)|)\leq \exp((|a|+1)|z|),
\end{split}
\end{equation}
which is the first chain of inequalities in \eqref{e1}. For any $N\in\mathbb N$, we have
\begin{equation}\label{e3}
\begin{split}
& \left| \cosh\left( \frac zN \right) - \cosh\left( a \frac zN \right) \right| \leq 2 \left| \sinh\left(\frac{(a-1)z}{2N} \right) \sinh\left( \frac{(a+1)z}{2N} \right) \right| \\
& \leq \frac{|a^2-1|}{2N^2} |z|^2\exp\left(\frac{|a-1|+|a+1|}{2N} |z|\right)\leq \frac{|a^2-1|}{2N^2} |z|^2 \exp\left( \frac{\alpha+1}{N} |z| \right)
\end{split}
\end{equation}
and
\begin{equation}\label{e4}
\begin{split}
& \left| a\sinh\left(  \frac zN \right)-\sinh\left( \frac {az}N \right) \right|=\left| \sum_{k=0}^\infty \frac{ 1}{(2k+1)!}(a-a^{2k+1})\left( \frac zN\right)^{2k+1} \right|\\
& =\frac{|a^2-1|}{N^2}|z|^2\left| \sum_{k=1}^\infty \frac{ 1}{(2k+1)!} \left( \sum_{\ell=0}^{k-1} a^{2\ell +1} \right) \left( \frac{z}{N} \right)^{2k-1} \right|\\
&\leq \frac{|a^2-1|}{2N^2}|z|^2 \sum_{k=1}^\infty \frac{ \alpha^{2k-1} }{(2k-1)!(2k+1)}\left( \frac{|z|}{N} \right)^{2k-1}\\
&\leq \frac{|a^2-1|}{6N^2}|z|^2 \sum_{k=1}^\infty \frac{ 1 }{(2k-1)! }\left( \frac{\alpha |z|}{N} \right)^{2k-1}\leq \frac{|a^2-1|}{6N^2} |z|^2 \exp\left( \frac\alpha N |z| \right).\\
\end{split}
\end{equation}
It follows from the identity $A^N-B^N=(A-B)\sum_{j=0}^{N-1}A^jB^{N-1-j}$, together with estimates \eqref{e2}, \eqref{e3} and \eqref{e4}, that for any $N\in\mathbb N$ and $z\in \mathbb C_J$,
\begin{equation}\nonumber
\begin{split}
&|F_N(z,a)-e^{az}|= \left| \left( F_N^{\frac 1N} (z,a)\right)^N-\left( e^{\frac {az}N} \right)^{N} \right| \\
&\left| \cosh\left( \frac zN \right) - \cosh\left( a\frac zN \right) + \left(a \sinh\left( \frac zN \right) - \sinh\left( a \frac zN \right)\right) \right|\\
&\times \sum_{k=0}^{N-1} |F_N(z,a)|^{k/N}\left| \exp \left (az\frac{N-1-k}{N}\right) \right|
\end{split}
\end{equation}
\begin{equation}
\begin{split}
& \leq \frac 23 \frac{|a^2-1|}{N^2} |z|^2 \exp\left( \frac{\alpha +1}{N} |z| \right) \sum_{k=0}^{N-1}\exp\left(k \left( \frac{|a|+1}{N} \right) |z| +\frac{N-1-k}{N}|a||z|\right)\\
&\leq \frac 23 \frac {|a^2-1|}{N} |z|^2 \exp((\alpha+1)|z|).
\end{split}
\end{equation}
The second inequality in \eqref{e1} is proved.
\end{proof}
As we have mentioned above, the notion of superoscillation can be extended to more general functions than the exponential function. In this case, we have the notion of supershift, which includes superoscillations as a particular case.
\begin{definition}\label{d4}
Let $|a|>1$ and
 $ h_{j}(N)\in [-1,1]$, $j\in\{0,\dots, N\}$,  $N\in\mathbb N_0$.
Let $G(\lambda)$ be an entire left slice hyperholomorphic function.
We say that the sequence
\begin{equation*}
F_N(x,a)=\sum_{j=0}^N Z_j(N,a)
G( x h_{j}(N)),
\end{equation*}
where  $\{ Z_j(N,a)\}_{j,N}\subseteq \mathbb R_n$, $j=0,\ldots ,N$, for $ N\in \mathbb{N}$,
admits the supershift property if
$$
\lim_{N\to \infty}F_N(x,a)=
G(x a),
$$
uniformly on any compact subset of $\mathbb R^{n+1}$.
\end{definition}

We now show how we can form infinite-order differential operators in this setting, starting from an entire function. We will then use this result to demonstrate the supershift property in the slice hyperholomorphic setting.

\begin{theorem}\label{bounded_operatorq4}
Let $G$ be entire left slice hyperholomorphic function whose series expansion at zero is given by
\begin{equation*}
G(\lambda)=\sum_{s=0}^\infty \lambda^{s} G_{s}.
\end{equation*}
then the operator
$$
 \mathcal V(x,\partial_{y_0})f(y):=\sum_{s=0}^\infty (\partial_{y_0}^s f(y)) (x)^s G_s
$$
is a family of bounded operators from $A_1$ to $A_1$ depending on the parameters $x$. If $x$ belongs to a compact subset $K\subseteq\mathbb R^{n+1}$, then the family is uniformly continuous with respect to $x\in K$.
\end{theorem}
\begin{proof}
Let $f\in A_1$. We consider its Taylor series $f(x)=\sum_{\ell=0}^{\infty} x^\ell a_\ell $. By Lemma \ref{taylor_coeff_estimate}, we have that there exist positive constants $C_f$ and $b_f$ such that
$$|a_\ell|\leq C_f\frac{b^\ell_f}{ \ell!}.$$
Moreover, let $\{f_m\}$ be a sequence in $A_1$. If $f_m\to 0$ in $A_1$ then we can choose $C_{f_m}$ and $b_{f_m}$ so that $C_{f_m}\to 0$ and $b_{f_m}$ is bounded.
We have
 \[
 \begin{split}
 \left| \mathcal V(x,\partial_{y_0})(f(y)) \right| & =\left|  \sum_{s=0}^\infty (\partial_{y_0}^s f(y))(x)^s G_s \right| \\
 & \leq  2^{n/2} \sum_{s=0}^\infty \sum_{\ell=s}^\infty  \frac{\ell!}{(\ell-s)!} |a_\ell|  |y|^{\ell-s} |x|^s  |G_s|\\
 &= 2^{n/2} \sum_{s=0}^\infty \sum_{\ell=0}^\infty  \frac{(\ell+s)!}{\ell!} |a_{\ell+s}|  |y|^\ell |x|^s |G_s|\\
 &\leq 2^{n/2} C_f \sum_{s=0}^\infty \sum_{\ell=0}^\infty  \frac{ (\ell+s)!}{(\ell+s)! \ell!} b_f^{\ell+s}  |y|^\ell |x|^s |G_s|\\
 &= 2^{n/2} C_f \sum_{\ell=0}^\infty \frac {b_f^\ell |y|^\ell}{\ell!} \sum_{s=0}^\infty |b_f x|^s |G_s|
 \end{split}
\]
 which implies $(\mathcal V(x, \partial_{y_0}) f_m(y))\to 0$ in $A_1$ if $f_m\to 0$ in $A_1$ since the last series is convergent due to the function $G$ being entire. Thus the uniform continuity with respect to $x\in K$ of the family of operators $\mathcal V(x,\partial_{y_0})$ is proved by the estimates
 $$
 \left| \mathcal V(x,\partial_{y_0})(f(y)) \right| \leq  2^{n/2}C_f e^{b_f|y|} \sum_{s=0}^\infty |b_f x|^s |G_s|.
 $$
\end{proof}
We now start from a superoscillatory function and associate a suitable function to it. We will show that for such a function, the supershift property holds by using the infinite-order differential operators obtained in the previous result.

\begin{theorem}\label{propm=4q}
Let $a\in\mathbb R$ and $Z_j(n,a)\in\mathbb R_n$ be coefficients such that the sequence
$$
f_N(x,a):=\sum_{j=0}^N e^{h_j(N)x} Z_j(N,a),\quad x\in\mathbb R^{n+1}
$$
converges to $e^{ax}$ in the space $A_1(\mathbb R^{n+1})$. Let moreover $G$ be an entire left slice hyperholomorphic function. The sequence of functions
$$
F_N(x,a):=\sum_{j=0}^N Z_j(N,a) G(h_j(N)x)
$$
converges to $G(xa)$ uniformly on any compact subset of $\mathbb R^{n+1}$.
\end{theorem}
\begin{proof}
Let $\sum_{s=0}^{\infty} \lambda^s G_s$ be the Taylor series of $G$. First we observe that
\[
\begin{split}
F_N( x,a)&=\sum_{j=0}^NZ_j(N,a)G(x h_j(N))\\
&= \sum_{j=0}^NZ_j (N,a)\sum_{s=0}^\infty (x h_j(N))^s G_s\\
&= \sum_{j=0}^N \sum_{s=0}^\infty (h_j(N))^s Z_j (N,a)  (x)^s G_s \\
&= \sum_{j=0}^N \sum_{s=0}^\infty \partial_{y_0}^s \left( e^{h_j(N)y} Z_j (n,a)\right)\Big|_{y=0}(x)^s G_s. \\
\end{split}
\]
We define the operator
$$
 \mathcal V(x,\partial_{y_0})f(y):=\sum_{s=0}^\infty (\partial_{y_0}^s f(y)) (x)^sg_s
$$
and, by Theorem \ref{bounded_operatorq4} we have that
\[
\begin{split}
& \lim_{N\to \infty} \sum_{j=0}^NZ_j(N,a)G(x h_j(N))=\lim_{N\to \infty} \sum_{j=0}^N  \mathcal V(x,\partial_{y_0})(e^{h_j(N)y} Z_j(N,a) )\Big |_{y=0}\\
&=\mathcal V(x,\partial_{y_0}) \left(\lim_{N\to \infty} \sum_{j=0}^N  e^{h_j(N)y} Z_j(N,a) \right) \Big |_{y=0}=\mathcal V(x,\partial_{y_0}) \left( e^{ay} \right ) \Big |_{y=0}\\
&= \sum_{s=0}^\infty  a^s \left( (x)^s g_s \right)=G(x a),
\end{split}
\]
uniformly on $x\in K$ for any compact subset $K\subset \mathbb R^{n+1}$.
\end{proof}

\subsection{Several variables in the slice hyperholomorphic setting}
By several variables in the slice hyperholomorphic setting, we mean that we encode in the paravector several functions that correspond to different frequencies. We will investigate what happens with this natural modification of the paravector.
\begin{definition}[Superoscillating sequence in the slice hyperholomorphic setting]\label{superoscill}
A sequence of the type
\begin{equation*}
F_N(x_0,\ldots ,x_n,a)=\sum_{j=0}^N Z_j(N,a)  e^{e_0 x_0 h_{j,0}(N) +e_1x_1 h_{j,1}(N)+\dots+e_n x_n h_{j,n}(N)},
\end{equation*}
where  $\{Z_j(N,a)\}_{j,N}$, for $j=0,\ldots ,N$ and  $ N\in \mathbb{N}$, is a  Cliffordian-valued sequence, is said to be {\em a superoscillating sequence} if
$$
\sup_{j=0,\ldots ,N} \  |h_{j,\ell}(N)|\leq 1 ,\ \ {\rm for} \ \ell=0,...,n,
$$
 and $F_n(x_0,\ldots ,x_n,a)$ converges uniformly on $x$ in any compact subset of $\mathbb R^{n+1}$ to
 $$e^{e_0 x_0 g_0(a) +e_1x_1 g_1(a)+\dots +e_nx_n g_n(a)},$$ where $a\in\mathbb R$, $g_\ell$'s are continuous functions of a real variable whose domain is $\mathbb R$ and $|g_\ell (a)|>1$ for  $\ell=1,\ldots ,d$.
\end{definition}

\begin{remark}
We can study a special case of Definition \ref{superoscill} where
$g_{k}:\mathbb R\to \mathbb R$ for $k=0,\dots, n$
are real analytic functions such that $|g_k(a)|\leq 1$ if $|a|\leq 1$, their Taylor series are: $$
g_{k}(\lambda) =\sum_{\ell=0}^\infty \lambda^\ell g_{k,\ell},\ \ \ \ \ h_{j,k}(N)=g_k(h_j(N))
$$
with $h_j(N)\in \mathbb R$ and $|h_j(N)|\leq 1$ for any $N\in\mathbb N_0$ and $j=0,\dots, N$. We say that
\begin{equation}\label{basic_sequence_sev_2}
F_N(x_0,\ldots ,x_n,a)=\sum_{j=0}^N Z_j(N,a)  e^{e_0 x_0 g_0(h_{j}(N)) +\dots +e_n x_n g_n(h_{j}(N))},
\end{equation}
is a superoscillating sequence if
$$\lim_{N\to+\infty} F_N(x_0,\dots,x_n,a)=e^{e_0 x_0 g_0(a) +\dots +e_n x_n g_n(a)}$$
uniformly on compact subset of $\mathbb R^{n+1}$ and $|g_k(a)|>|g_k(h_{j}(N))|$ for any $k=0,\dots, n$ and $j=0,\dots,N$. Moreover, we observe that using the Taylor series of $g_k$'s we can rewrite $F_N(x_0,\dots,x_n,a)$ in the following way
$$ F_N(x_0,\ldots ,x_n,a)=\sum_{j=0}^N Z_j(N,a)  e^{\sum_{\ell=0}^\infty \left( \sum_{k=0}^n e_k x_k g_{k,\ell}\right) (h_{j}(N))^\ell }. $$
We also define $$u_\ell(x_0,\dots, x_n):=\sum_{k=0}^n e_k x_k g_{k,\ell}.$$
\end{remark}
\begin{theorem}\label{bounded_operator}
Let $\left\{u_\ell \right\}_{\ell\in\mathbb N_0}$ be a sequence in $\mathbb R^{n+1}$ such that
$$
\lim \sup_{\ell\to\infty}|u_{\ell} |^{1/\ell}= 0,
$$
then the operator
$$\mathcal U(\partial_{x_0}) \left( f( x ) \right) :=\sum_{s=0}^\infty \partial_{x_0}^s (f ( x))\sum_{k=0}^\infty \frac 1{k!}\sum_{t_1+\dots+t_k=s} u_{t_1} \cdots u_{t_k} $$
is a continuous operator from $A_1$ to $A_1$.
\end{theorem}
\begin{proof}
Let $f\in A_1$. We consider its Taylor series: $f(x)=\sum_{\ell=0}^{\infty} x^\ell a_\ell $. In particular, by Lemma 2.13 in \cite{CPSS21} we have that there exist constants $C_f$ and $b_f$ such that
$$|a_\ell|\leq C_f\frac{b^\ell_f}{ \ell!}.$$
Moreover, if $f_m\to 0$ in $A_1$ then $C_{f_m}\to 0$ and $b_{f_m}$ is bounded.
We have
 \[
 \begin{split}
 \left| \mathcal U(\partial_{x_0})(f(x)) \right| & =\left| \sum_{k=0}^\infty \frac 1{k!} \sum_{s=0}^\infty (\partial_{x_0}^s f(x)) \sum_{t_1+\dots+t_k=s} u_{t_1}\cdots u_{t_k} \right| \\
 & \leq \sum_{s=0}^\infty \sum_{\ell=s}^\infty  \frac{\ell!}{(\ell-s)!} |a_\ell|  |x|^{\ell-s} \left( \sum_{k=0}^\infty \frac 1{k!} \sum_{t_1+\dots +t_k=s} |u_{t_1} | \cdot\dots\cdot |u_{t_k} |\right)\\
 &= \sum_{s=0}^\infty \sum_{\ell=0}^\infty  \frac{(\ell+s)!}{\ell!} |a_{\ell+s}|  |x|^\ell \left( \sum_{k=0}^\infty \frac 1{k!} \sum_{t_1+\dots +t_k=s} |u_{t_1}| \cdot\dots\cdot |u_{t_k}|\right)\\
 &\leq C_f \sum_{s=0}^\infty \sum_{\ell=0}^\infty  \frac{ (\ell+s)!}{(\ell+s)! \ell!} b_f^{\ell+s}  |x|^\ell \left( \sum_{k=0}^\infty \frac 1{k!} \sum_{t_1+\dots +t_k=s} |u_{t_1}| \cdot\dots\cdot |u_{t_k}|\right)\\
 &= C_f \sum_{\ell=0}^\infty \frac {b_f^\ell |x|^\ell}{\ell!} \sum_{s=0}^\infty \sum_{k=0}^\infty \frac 1{k!} \sum_{t_1+\dots +t_k=s} b_f^{t_1}| u_{t_1}| \cdot\dots\cdot b_f^{t_k}| u_{t_k}|\\
 & = C_f e^{b_f |x|} \sum_{k=0}^{\infty} \frac 1{k!} \left( \sum_{t=0}^\infty  b_f^t |u_t| \right)^k\\
 & = C_f e^{b_f |x|} e^{\sum_{t=0}^\infty b_f^t |u_t |}
 \end{split}
\]
which implies, by the Weierstrass criterion,  that if $f\in A_1$ then $\mathcal U(\partial_{x_0})(f(x))$ is a series of slice hyperhomorphic functions converging with respect the sup-norm on any compact subset of $\mathbb R^{n+1}$ to a slice hyperholomorphic function. Moreover, the same estimate implies that if $\{f_m\}$ is a sequence in $A_1$ converging to zero then $\mathcal U(\partial_{x_0})(f_m(x)) \to 0$ in $A_1$.
\end{proof}

\begin{theorem}\label{conv_superos}
Let $\left\{u_\ell: \mathbb R^{n+1}\to \mathbb R^{n+1}\right\}_{\ell\in\mathbb N_0}$ be a sequence of continuous functions such that
$$
\lim \sup_{\ell\to\infty} \left( \sup_{\bx\in K} |u_{\ell} (x)|\right)^{1/\ell} = 0\quad \textrm{for any compact subsets $K\subset \mathbb R^{n+1}$}.
$$
Let $a\in\mathbb R$ and
\begin{equation}\label{basic_sequence_bis}
f_N (x,a):=\sum_{j=0}^N e^{h_j(N)x}Z_j(N,a),
\end{equation}
be a sequence of function such that $|h_j(N)|\leq 1$ and it converges to $e^{ax}$ in $A_1$.
Let also
$$
F_N(x,a):=\sum_{j=0}^N Z_j(N,a)e^{\sum_{\ell=0}^\infty u_\ell(x) h_j(N)^\ell}.
$$
Then we have
$$
\lim_{N\to \infty}F_N(x,a)=e^{\sum_{\ell=0}^\infty u_\ell(x) a^\ell},
$$
uniformly on compact subsets of $\mathbb R^{n+2}$.
\end{theorem}

\begin{proof}
We want to prove that
$$ \lim_{N\to+\infty} \sum_{j=0}^N Z_j(N,a) e^{\sum_{\ell=0}^\infty u_\ell( x) h_j(N)^\ell}=e^{\sum_{\ell=0}^\infty u_\ell(x) a^\ell},$$
uniformly on any compact subset of $\mathbb R^{n+2}$. We fix a compact subset $K$ of $\mathbb R^{n+2}$.
We observe that for any $(x,a)\in K$ we have
\[
\begin{split}
& \sum_{j=0}^N Z_j(N,a) e^{\sum_{\ell=0}^\infty u_\ell(x) h_j(N)^\ell}  = \sum_{j=0}^N Z_j(N,a) \sum_{k=0}^\infty \frac{\left( \sum_{\ell=0}^\infty u_\ell(x) h_j(N)^\ell\right)^k}{k!}\\
& =\sum_{j=0}^N Z_j(N,a) \sum_{k=0}^\infty \frac 1{k!}\sum_{s=0}^\infty \left(\sum_{t_1+\dots +t_k=s} u_{t_1}(x)\cdot \dots \cdot u_{t_k}(x) \right)h_j(N)^s\\
& =\sum_{j=0}^N  \sum_{k=0}^\infty \frac 1{k!}\sum_{s=0}^\infty h_j(N)^s Z_j(N,a) \left(\sum_{t_1+\dots +t_k=s} u_{t_1}(x)\cdot \dots \cdot u_{t_k}(x) \right)\\
& =\sum_{j=0}^N  \sum_{k=0}^\infty \frac 1{k!}\sum_{s=0}^\infty \partial_{y_0}^s\left( e^{ h_j(N) y} Z_j(N,a)\right) \big |_{y=0} \left(\sum_{t_1+\dots +t_k=s} u_{t_1}(x)\cdot \dots \cdot u_{t_k}(x) \right).
\end{split}
\]
We define the operator
$$
\mathcal U(x,\partial_{y_0}):=\sum_{k=0}^\infty \frac 1{k!}\sum_{s=0}^\infty (\partial_{y_0}^s) \left(\sum_{t_1+\dots +t_k=s} u_{t_1}(x)\cdot \dots \cdot u_{t_k}(x)\right)
$$
and, since by Theorem \ref{bounded_operator} the operator  $\mathcal U(x,\partial_{y_0})$ is a uniformly  bounded operator from $A_1$ to $A_1$ with respect to $(x,a)\in K$ because of the estimate: $$\left| \mathcal U(x,\partial_{y_0})(f(y)) \right|\leq C_f e^{b_f |y|} e^{\sum_{t=0}^\infty b_f |u_t(x) |},$$ we have that
\[
\begin{split}
\lim_{N\to \infty} \sum_{j=0}^N Z_j(N,a) e^{\sum_{\ell=0}^\infty u_\ell(x) h_j(N)^\ell}&=\lim_{N\to \infty} \sum_{j=0}^N  \mathcal U(x,\partial_{y_0})(e^{h_j(N)y} Z_j(N,a) )\big |_{y=0}\\
&=\mathcal U(x,\partial_{y_0}) \left(\lim_{N\to \infty} \sum_{j=0}^N  e^{h_j(N)y} Z_j(N,a) \right) \big |_{y=0}\\
&=\mathcal U(x,\partial_{y_0}) \left( e^{ay} \right ) \big |_{y=0}\\
&= \sum_{k=0}^\infty \frac 1{k!}\sum_{s=0}^\infty  a^s \left(\sum_{t_1+\dots +t_k=s} u_{t_1}(x)\cdot \dots \cdot u_{t_k}(x)\right)\\
&=\sum_{k=0}^\infty \frac 1{k!} \left(\sum_{\ell =0}^\infty  u_\ell (x) a^\ell \right)^k\\
&=e^{\sum_{\ell=0}^\infty u_\ell(x) a^\ell}
\end{split}
\]
where the convergence of the last series to the exponential function is uniform on the compact subset $K$.
\end{proof}

\begin{corollary}\label{cor1}
Let $g_k(a):=\sum_{\ell=0}^\infty g_{k,\ell}a^\ell$ for $k=0,\dots, n$ be real analytic functions of real variable such that $|g_k(a)|>1$ for $|a|>1$ and $|g_k(a)|\leq 1$ for $|a|\leq1$, if we suppose that the sequence $$f_N(x,a)=\sum_{j=0}^N  e^{h_j(N) x} Z_j(N,a) $$ converges to $e^{ax}$ in $A_{1}$,  then the sequence
\[
\begin{split}
F_N(x,a)& =\sum_{j=0}^N Z_j(N,a) e^{e_0 x_0 g_{0}(h_j(N))+e_1 x_1 g_1(h_j(N))+\dots +e_n x_n g_n(h_j(N))}
\end{split}
\]
converges uniformly on compact subset of $\mathbb R^{n+2}$ to $e^{e_0 x_0g_0(a)+\dots+ e_n x_ng_n(a)}$ and thus it is a superoscillating sequence acording to Definition \ref{superoscill}.
\end{corollary}
\begin{proof}
The proof is a direct consequence of the Theorem \ref{conv_superos} since the functions $g_k$ are entire functions, we have
\[
\begin{split}
\limsup_{\ell\to\infty} \sup_{x\in K} |u_{\ell} (x)|^{1/\ell} & \leq  \limsup_{\ell\to\infty} \sup_{x\in K} \sum_{k=0}^n |x_k|^{1/\ell} |g_{k,\ell}|^{1/\ell} \\
& \leq   \limsup_{\ell\to\infty} \sup_{x\in K} \left( \sum_{k=0}^n |x_k|^{1/\ell}\right ) \sum_{k=0}^n |g_{k,\ell}|^{1/\ell} =0
\end{split}
\]
for any compact subset $K\subset \mathbb R^{n+1}$.
\end{proof}

%Comincia la parte bis da scegliere

\begin{definition}[Supershifts in the slice hyperholomorphic setting] \label{ss1}
Let $a\in\mathbb R$. Let
 $\{ h_{j,\ell}(N) \}$,  $j=0,...,N$  for $ N\in \mathbb{N}$, be  real-valued sequences for $\ell=0,...,n$
 such that
$$
\sup_{j=0,\ldots ,N,\ N\in\mathbb{N}} \  |h_{j,\ell}(N)|\leq 1 ,\ \ {\rm for} \ \ell=0,...,n.
$$
Let $G(\lambda)$, be an entire left slice hyperholomorphic function.
We say that
the sequence
\begin{equation}\label{basic_sequence_sevFF_5}
F_N(x,a)=\sum_{j=0}^N Z_j(N,a)
G( x_0 h_{j,0}(N) +e_1x_1 h_{j,1}(N)+\dots +e_n x_n h_{j,n}(N)),
\end{equation}
where  $\{ Z_j(N,a)\}_{j,N}$, $j=0,\ldots ,N$, for $ N\in \mathbb{N}$ is a  Clifford-valued sequence,
admits the supershift property if
$$
\lim_{N\to \infty}F_N(x,a)=
G(x_0g_0(a) +e_1x_1g_1(a)+\dots +e_n x_n g_n(a)),
$$
where $g_0,\dots, g_n$ are continuous functions such that $|g_j(a)|>1$, uniformly on $x$ in any compact subset $K\subset \mathbb R^{n+1}$.
\end{definition}

\begin{theorem}\label{bounded_operatorq5}
Let $G$ be an entire left slice hyperholomorphic functions whose series expansions at zero is given by
\begin{equation*} 
G (\lambda)=\sum_{m=0}^\infty \lambda^{m} G_{m}.
\end{equation*}
and $g_0,\dots, g_n$ be real analytic on $\mathbb R$ and real valued functions then the operators
$$
 \mathcal V(x,\partial_{y_0}) f(y):=\sum_{s=0}^\infty \left(  \sum_{t=0}^{\infty}  \partial_{y_0}^t \left( f(y) \right) \sum_{\ell_1+\dots+\ell_s= t} u_{\ell_1}(x)\cdots u_{\ell_s}(x)\right) G_s
 $$
where $u_\ell(x):= x_0 g_{0,\ell}+e_1 x_1 g_{1,\ell} +\dots + e_n x_n g_{n,\ell}$ and $g_{j,\ell}$ are the Taylor coefficients of $G_i$ for $i=0,\dots, n$ and $\ell\in\mathbb N_0$, form a family of bounded operators from $A_1$ to $A_1$ depending on the parameters $x$. If $x$ belongs to a compact subset $K\subset\mathbb R^{n+1}$, then the family is uniformly bounded with respect to $x\in K$.
\end{theorem}
\begin{proof}
Let $f\in A_1$. We consider its Taylor series: $f(x)=\sum_{\ell=0}^{\infty} x^\ell a_\ell $. In particular, by Lemma \ref{taylor_coeff_estimate} we have that there exist constants $C_f$ and $b_f$ such that
$$|a_\ell|\leq C_f\frac{b^\ell_f}{ \ell!}.$$
 Moreover, if $f_m\to 0$ in $A_1$ then we can choose $C_{f_m}$ and $b_{f_m}$ so that $C_{f_m}\to 0$ and $b_{f_m}$ is bounded.
 Since the functions $g_j$'s are real analytic in all the real line and the function $G$ is entire slice hyperholomorphic, for any $K\subset \mathbb R^{n+1}$ we have
 \begin{equation}\label{convuell}
 \begin{split}
 \limsup_{\ell\to\infty} \sup_{x\in K} |u_\ell(x)|^{\frac 1\ell}& \leq \limsup_{\ell\to\infty} \sup_{x\in K} \left( |x_0 g_{0,\ell}|^{\frac 1\ell}+|x_1 e_1 g_{1,\ell}|^{\frac 1\ell}+\dots+|x_n e_n g_{n,\ell}|^{\frac 1\ell} \right) \\
 & \leq  \limsup_{\ell\to\infty} (\sup_{x\in K} |x|)^{\frac 1\ell} \left( |g_{0,\ell}|^{\frac 1\ell}+|e_1 g_{1,\ell}|^{\frac 1\ell}+\dots+|e_n g_{n,\ell}|^{\frac 1\ell}\right) = 0
 \end{split}
 \end{equation}
 and
 \begin{equation}\label{convg}
 \limsup_{s\to\infty} |G_s|^{\frac 1s}=0 .
 \end{equation}
 We observe that
we have
 \begin{equation}\label{estimate_cont}
 \begin{split}
 & \left| \mathcal V(x,\partial_{y_0})(f(y)) \right|  =\left|  \sum_{s=0}^\infty \left(  \sum_{t=0}^{\infty}  \partial_{y_0}^t \left( f(y) \right) \sum_{\ell_1+\dots+\ell_s= t} u_{\ell_1}(x)\cdots u_{\ell_s}(x)\right) G_s\right| \\
 & \leq 2^{n/2}  \sum_{s=0}^\infty \left( \sum_{t=0}^\infty \sum_{\ell=t}^\infty  \frac{\ell!}{(\ell-t)!} |a_\ell|  |y|^{\ell-t}  \sum_{\ell_1+\dots+\ell_s= t} |u_{\ell_1}(x)|\cdots |u_{\ell_s}(x)| \right) |G_s|\\
 &= 2^{n/2} \sum_{s=0}^\infty\left( \sum_{t=0}^\infty \sum_{\ell=0}^\infty  \frac{(\ell+t)!}{\ell!} |a_{\ell+t}|  |y|^\ell \left( \sum_{\ell_1+\dots+\ell_s= t} |u_{\ell_1}(x)|\cdots |u_{\ell_s}(x)| \right)\right) |G_s|\\
 &\leq 2^{n/2} C_f \sum_{s=0}^\infty \left(  \sum_{t=0}^\infty \sum_{\ell=0}^\infty  \frac{ (\ell+t)!}{(\ell+t)! \ell!} b_f^{\ell+t}  |y|^\ell \left( \sum_{\ell_1+\dots+\ell_s= t} |u_{\ell_1}(x)|\cdots |u_{\ell_s}(x)| \right) \right) |G_s|\\
 &= 2^{n/2}  C_f \sum_{\ell=0}^\infty \frac {b_f^\ell |y|^\ell}{\ell!} \sum_{s=0}^\infty \left( \sum_{t=0}^{\infty}\sum_{\ell_1+\dots +\ell_s=t} |(b_f)^{\ell_1} u_{\ell_1}(x)|\cdots |(b_f)^{\ell_n} u_{\ell_s}(x)| \right)|G_s|\\
 & = 2^{n/2} C_f e^{b_f |y|} \sum_{s=0}^\infty \left( \sum_{\ell=0}^\infty |b_f|^\ell |u_{\ell}(x)| \right)^s |G_s|
 \end{split}
\end{equation}
by \eqref{convuell} and \eqref{convg}, we can conclude that all the series in the last term are convergent and, moreover, $(\mathcal V(x,\partial_{y_0})f_m(y))\to 0$ in $A_1$ if $f_m\to 0$ in $A_1$. In particular, if we fix a compact subset $K\subset \mathbb R^{n+1}$, by the estimate \eqref{estimate_cont},
we can also deduce that the family of continuous operators $\mathcal V(x,\partial_{y_0})$ parametrized by $x\in K$ is indeed uniformly continuous.
\end{proof}

\begin{theorem}\label{propm=3q}
Let $a\in\mathbb R$ and let
\begin{equation} \nonumber
f_N(x,a):= \sum_{j=0}^N e^{h_j(N) x} Z_j(N,a),\ \ \ N\in \mathbb{N},\ \ \ (x_1,\dots ,x_n)\in \mathbb{R}^n,
\end{equation}
be a sequence of slice hyperholomorphic functions converging to $e^{ax}$ in the space $A_{1}$.
Let $G$ be an entire left slice hyperholomorphic functions.
Let $g_0,\dots, g_n$ be real analytic functions on $\mathbb R$. We define
$$
F_N(x_0,\dots ,x_n,a)=\sum_{j=0}^N Z_j(N,a) G( x_0 g_0 (h_j(N)) +e_1x_1 g_1(h_{j}(N))+\dots +e_n x_n g_n(h_{j}(N))).
$$
 If $|g_j(a)|>1$ for $|a|>1$ and $|g_j(a)|\leq 1$ for $|a|\leq 1$ then $F_N(x_0,\dots,\,x_n,a)$ admits the supershift property that is
$$
\lim_{N\to \infty}F_N(x,a)=G( x_0 g_0 (a) +e_1x_1 g_1(a)+\dots +e_n x_n g_n(a)),
$$
uniformly on any compact subsets of $\mathbb R^{n+1}$.
\end{theorem}
\begin{proof}
The Taylor series of $G$ is given by
\begin{equation*}
G (\lambda)=\sum_{m =0}^\infty \lambda^{m} G_{ m}.
\end{equation*}
First we observe that
\[
\begin{split}
F_N(x,a)&=\sum_{j=0}^NZ_j(N,a)G( x_0 g_0 (h_j(N)) +e_1x_1 g_1(h_{j}(N))+\dots +e_n x_n g_n(h_{j}(N)))\\
&= \sum_{j=0}^NZ_j (N,a)\sum_{s=0}^\infty \left( x_0 g_0 (h_j(N)) +e_1x_1 g_1(h_{j}(N))+\dots +e_n x_n g_n(h_{j}(N))\right)^s G_s\\
&= \sum_{j=0}^N \sum_{s=0}^\infty  Z_j (N,a) \left( \sum_{\ell=0}^{\infty} \left(x_0 g_{0,\ell}+e_1 x_1 g_{1,\ell} +\dots + e_n x_n g_{n,\ell} \right) h_j(N)^{\ell} \right)^s G_s \\
\end{split}
\]
We define $u_\ell(x):= x_0 g_{0,\ell}+e_1 x_1 g_{1,\ell} +\dots + e_n x_n g_{n,\ell}$ and we observe that
$$
\left( \sum_{\ell=0}^{\infty} u_{\ell}(x) h_j(N)^{\ell} \right)^s=\sum_{t=0}^{\infty}\sum_{\ell_1+\dots+\ell_s= t} u_{\ell_1}(x)\cdots u_{\ell_s}(x) (h_j(N))^t.
$$
 Thus we have
 \[
 \begin{split}
 F_N(x,a)&=\sum_{j=0}^N \sum_{s=0}^\infty  Z_j (N,a) \left( \sum_{t=0}^{\infty}\sum_{\ell_1+\dots+\ell_s= t} u_{\ell_1}(x)\cdots u_{\ell_s}(x) (h_j(N))^t \right) G_s\\
 &= \sum_{s=0}^\infty \left(  \sum_{t=0}^{\infty} \sum_{j=0}^N (h_j(N))^t Z_j (N,a) \sum_{\ell_1+\dots+\ell_s= t} u_{\ell_1}(x)\cdots u_{\ell_s}(x)\right) G_s\\
 &=\sum_{s=0}^\infty \left(  \sum_{t=0}^{\infty}  \partial_{y_0}^t\left(  \sum_{j=0}^N e^{h_j(N)y} Z_j (N,a)\right)\Big|_{y=0} \sum_{\ell_1+\dots+\ell_s= t} u_{\ell_1}(x)\cdots u_{\ell_s}(x)\right) G_s.
\end{split}
\]
Since the operator
$$
\mathcal V(x,\partial_{y_0}):=\sum_{s=0}^\infty \left(  \sum_{t=0}^{\infty}  \partial_{y_0}^t \left( \cdot \right) \sum_{\ell_1+\dots+\ell_s= t} u_{\ell_1}(x)\cdots u_{\ell_s}(x)\right) G_s
$$
is bounded from $A_1$ to $A_1$, by Theorem \ref{bounded_operatorq5}   we have that
\[
\begin{split}
& \lim_{N\to \infty} \sum_{j=0}^NZ_j(N,a)G(x_0 g_0 (h_j(N)) +e_1x_1 g_1(h_{j}(N))+\dots +e_n x_n g_n(h_{j}(N)))\\
&=\lim_{N\to \infty} \sum_{j=0}^N  \mathcal V(x,\partial_{y_0})(e^{h_j(N)y} Z_j(N,a) )\Big |_{y=0}
\end{split}
\]
\[
\begin{split}
&=\mathcal V(x,\partial_{y_0}) \left(\lim_{N\to \infty} \sum_{j=0}^N  e^{h_j(N)y} Z_j(N,a) \right) \Big |_{y=0}=\mathcal V(x,\partial_{y_0}) \left( e^{ay} \right ) \Big |_{y=0}\\
&= \sum_{s=0}^\infty \left(  \sum_{t=0}^{\infty}  a^t \sum_{\ell_1+\dots+\ell_s= t} u_{\ell_1}(x)\cdots u_{\ell_s}(x)\right) G_s\\
&=G(x_0 g_0 (a) +e_1x_1 g_1(a)+\dots +e_n x_n g_n(a)),
\end{split}
\]
uniformly on any compact subset of $\mathbb R^{n+1}$.
\end{proof}

\section{Monogenic setting}\label{smonogenic}

We now consider the monogenic setting, where we start with the prototype of a superoscillating sequence in one Clifford variable. Then we examine the supershift in one Clifford variable and conclude with the case of several Clifford variables.

\subsection{The prototype of a superoscillating sequence in one Clifford variable.}
The prototype of a superoscillating sequence is given by
$$
F_{N}(x_1,\dots, x_n,a):=\left( \cosh(\underline{x} /N)+a\sinh(\underline {x}/N) \right)^{N}=\sum_{j=0}^NC_j(N,a) e^{(1-(2j)/N)\underline x}.
$$
Observe that, as for the slice hyperholomorphic case, we have
$$
\lim_{N\to+\infty} F_N(x_1,\dots, x_n,a)=e^{a\underline x}.
$$
In a similar way, as we did for the slice hyperholomorphic case, we can extend $F_N(x_1,\dots, x_n,a)$ to an entire monogenic function in the following way
\begin{equation}\label{proto_superos}
F_N(x,a):=\left( \operatorname{Cosh}(x /N)+a\operatorname{Sinh}( x/N) \right)^{\odot_L N}=\sum_{j=0}^NC_j(N,a) E((1-(2j)/N) x).
\end{equation}

\begin{lemma}\label{l2}
Let $a\in\mathbb R$ with $\alpha:=\max(1,|a|)$ and, for any $x\in\mathbb R^{n+1}$, let $F_N(x,a)$ be defined as in \eqref{proto_superos}. For any $x\in\mathbb R^{n+1}$, one has
\begin{equation}\label{e1_bis}
\begin{split}
& |F_N(x,a)|\leq \exp((\alpha+1)n |x|)\\
& F_N(x,a)\to E(xa)\textrm{ in }  A_1 \textrm{ when } N\to+\infty.
\end{split}
\end{equation}
\end{lemma}
\begin{proof}
We observe that
$$
\operatorname{sinhc}(|z_1|+\dots +|z_n|):=\frac {\sinh(|z_1|+\dots +|z_n|)}{|z_1|+\dots +|z_n|}=\int_0^1 \cosh(t(|z_1|+\dots +|z_n|))\, dt$$ and, we have the estimate $|\operatorname{sinhc}(|z_1|+\dots +|z_n|)|\leq e^{ n|x|}$. Thus, for any $N\in\mathbb N$ and for any $x\in \mathbb R^{n+1}$ we have by \eqref{hyper_e2}
\begin{equation}\label{e2_bis}
\begin{split}
& |F_N(x,a)|=\left|\left( \operatorname{Cosh}\left(\frac xN \right)+a\operatorname{Sinh}\left(\frac xN\right) \right)^{\odot_L N} \right|\\
& \leq \left( \cosh\left(\frac {|z_1'|+\dots +|z_n'|}N \right)-\frac {|z_1'|+\dots +|z_n'|}N a\operatorname{sinhc}\left(\frac {|z_1'|+\dots +|z_n'|}N\right) \right)^N\\
& \leq e^{n |x|}\left( 1+\frac {|a(|z_1'|+\dots +|z_n'|)|}{N} \right)^N\leq \exp((|a|+1)n |x|),
\end{split}
\end{equation}
which is the first inequality in \eqref{e1_bis}. Note that, for any positive real values $a,\, b\in\mathbb R$ we have, by the uniqueness of the CK-extension and since the next equality easily holds true for $x$ replaced by $\unx$:
 $$\operatorname{Cosh}(ax)+\operatorname{Cosh}(bx)=2 \operatorname{Sinh}\left ((a-b) \frac x2\right)\odot_L  \operatorname{Sinh} \left((a+b)\frac x2\right).
$$ Thus for any $N\in\mathbb N$, we have
\begin{equation}\label{e3_bis}
\begin{split}
& \left| \operatorname{Cosh}\left( \frac xN \right) - \operatorname{Cosh}\left( a \frac xN \right) \right| =2 \left| \operatorname{Sinh}\left(\frac{(a-1)x}{2N} \right)\odot_L \operatorname{Sinh}\left( \frac{(a+1)x}{2N} \right) \right| \\
& =2\left| \sum_{k=1}^{\infty} \sum_{j=1}^\infty \frac{(a-1)^{2k+1}(a+1)^{2j+1}(z_1 e_1+\cdots+z_ne_n)^{\odot_L 2(j+k)+2}}{(2j+1)!(2k+1)! (2N)^{2j+2k+2}}\right| \\
& \overset{\eqref{hyper_e1}}{\leq} 2 \sum_{k=1}^{\infty} \sum_{j=1}^\infty \frac{|a-1|^{2k+1} |a+1|^{2j+1}(|z_1|+\cdots+|z_n|)^{2(j+k)+2}}{(2j+1)!(2k+1)! (2N)^{2j+2k+2}}\\
& \leq \frac{|a^2-1|}{2N^2} (|z_1|+\dots+|z_n|)^{2}\exp\left(\frac{|a-1|+|a+1|}{2N} (|z_1|+\cdots+|z_n|)\right)\\
&\leq n^2 \frac{|a^2-1|}{2N^2} |x|^2 \exp\left( n \frac{\alpha+1}{N} |x| \right)
\end{split}
\end{equation}
and
\begin{align} \label{e4_bis} \nonumber
& \left| a\operatorname{Sinh}\left(  \frac xN \right)-\operatorname{Sinh}\left( \frac {ax}N \right) \right|=\left| \sum_{k=0}^\infty \frac{ (a-a^{2k+1})}{(2k+1)!}\left( \frac {z_1e_1+\dots+z_n e_n}N\right)^{\odot_L 2k+1} \right| \nonumber \\
&\overset{\eqref{hyper_e1}}{\leq} \frac{|a^2-1|}{N^2}(|z_1|+\dots+|z_n|)^2\left| \sum_{k=1}^\infty \frac{1}{(2k+1)!} \left( \sum_{\ell=0}^{k-1} a^{2\ell +1} \right) \left( \frac{|z_1|+\dots+|z_n|}{N} \right)^{2k-1} \right| \nonumber \\
& \leq \frac{|a^2-1|}{2N^2}(|z_1'|+\dots+|z_n'|)^2 \sum_{k=1}^\infty \frac{ \alpha^{2k-1} }{(2k-1)!(2k+1)}\left( \frac{ |z_1|+\dots+|z_n|}{N} \right)^{2k-1}\\
&\leq \frac{|a^2-1|}{6N^2} (|z_1|+\dots+|z_n|)^2 \sum_{k=1}^\infty \frac{ 1 }{(2k-1)! }\left( \frac{\alpha (|z_1|+\dots+|z_n|)}{N} \right)^{2k-1} \nonumber\\
&\leq \frac{|a^2-1|}{6N^2} n^2|x|^2 \exp\left( \frac\alpha N (|z_1|+\dots+|z_n|) \right)\leq \frac{|a^2-1|}{6N^2} n^2|x|^2 \exp\left( n\frac\alpha N |x| \right) \nonumber
\end{align}
 From the identity:
 $$(f(x))^{\odot_L N}-(g(x))^{\odot_L N}=(f(x)-g(x))\odot_L \sum_{k=0}^{N-1}(f(x))^{\odot_L k}(g(x))^{\odot_L N-1-k}$$ when $f(x)$ and $g(x)$ are entire monogenic functions that commute we have that
 \[
 \begin{split}
 F_N(x,a)-E(xa)&=\left(\operatorname{Cosh}\left(\frac xN\right)+a\operatorname{Sinh} \left(\frac xN\right)\right)^{\odot_L N}-\left(E\left( a \frac xN\right)\right)^{\odot_L N}\\
 &=\left( \operatorname{Cosh} \left( \frac xN \right) - \operatorname{Cosh} \left( a\frac xN \right) + a \operatorname{Sinh}\left( \frac xN \right)\right. \\
&\left. - \operatorname{Sinh}\left( a \frac xN \right)\right) \odot_L \sum_{k=0}^{N-1} \left(\operatorname{Cosh}\left(\frac xN\right)+a\operatorname{Sinh} \left(\frac xN\right) \right)^{\odot_L k} \odot_L E \left (ax\frac{N-1-k}{N}\right)
 \end{split}
 \]
 . Moreover, by the equations:
 $$F_N(x,a)\odot_L E(ax)=E(ax) \odot_L F_N(x,a),$$
  $$\operatorname{Sinh}(x)\odot_L E(ax)= E(ax) \odot_L \operatorname{Sinh}(x)$$ and $$\operatorname{Cosh}(x)\odot_L E(ax)= E(ax) \odot_L \operatorname{Cosh}(x)$$ (that is true because of the uniqueness of the CK-extension and the fact that $F_N(\unx,a) $,  $ e^{a\unx}$, $\operatorname{Sinh}(x)$ and $\operatorname{Cosh}(x)$ commute with each other), and Lemma \ref{l12} that for any $N\in\mathbb N$, $\eta>0$, $\delta>0$, there exists $C(n,\delta,\eta)>0$ such that
\begin{eqnarray}\label{estimate_cont_mono}
& \|F_N(x,a)-E(xa)\|_{1, n(1+\eta)(n+\delta+1)(2n(\alpha+1)+1)}\leq C(n,\delta,\eta) \left \| \left( \operatorname{Cosh} \left( \frac xN \right) - \operatorname{Cosh} \left( a\frac xN \right) \right. \right. \nonumber \\
&\left. \left. + a \operatorname{Sinh}\left( \frac xN \right) - \operatorname{Sinh}\left( a \frac xN \right)\right) \odot_L \sum_{k=0}^{N-1} F_N(z,a)^{\odot_L k/N} \odot_L E \left (ax\frac{N-1-k}{N}\right) \right\|_{1, (1+\eta)(n+\delta+1)(2n(\alpha+1)+1)} \nonumber \\
&\leq C(n,\delta,\eta) \left \| \operatorname{Cosh} \left( \frac xN \right) - \operatorname{Cosh} \left( a\frac xN \right)+ a \operatorname{Sinh}\left( \frac xN \right) - \operatorname{Sinh}\left( a \frac xN \right)\right\|_{1,1} \nonumber \\
&\times \left\| \sum_{k=0}^{N-1} \left(\operatorname{Cosh}\left(\frac xN\right)+a\operatorname{Sinh} \left(\frac xN\right) \right)^{\odot_L k} \odot_L E \left (ax\frac{N-1-k}{N}\right) \right\|_{1,2n(\alpha+1)} \\
&\leq C(n,\delta,\eta) \left \| \operatorname{Cosh} \left( \frac xN \right) - \operatorname{Cosh} \left( a\frac xN \right)+ a \operatorname{Sinh}\left( \frac xN \right) - \operatorname{Sinh}\left( a \frac xN \right)\right\|_{1,1} \nonumber \\
&\times  \sum_{k=0}^{N-1} \sup_{x\in\mathbb R^{n+1}} \left( \left| \left(\operatorname{Cosh}\left(\frac xN\right)+a\operatorname{Sinh} \left(\frac xN\right) \right)^{\odot_L k} \odot_L E \left (ax\frac{N-1-k}{N}\right) \right| \exp(-2n(\alpha+1) |x|) \right) \nonumber \\
&\overset {\eqref{hyper_e2}}{\leq} C(n,\delta,\eta) \left \| \operatorname{Cosh} \left( \frac xN \right) - \operatorname{Cosh} \left( a\frac xN \right)+ a \operatorname{Sinh}\left( \frac xN \right) - \operatorname{Sinh}\left( a \frac xN \right)\right\|_{1,1}\nonumber \\
&\times  \sum_{k=0}^{N-1} \sup_{x\in\mathbb R^{n+1}} \left( \left(\operatorname{cosh}\left(\frac {|z_1|+\dots+|z_n|}N\right)+a\operatorname{sinh} \left(\frac {|z_1|+\dots+|z_n|}N\right) \right)^{k}  \right. \nonumber\\
&\left. \times \exp\left( \left (a(|z_1|+\dots+|z_n|)\frac{N-1-k}{N}\right) \right) \exp(-2n(\alpha+1) |x|) \right) \nonumber \\
&\overset{ \textrm{\eqref{e2_bis}, \eqref{e3_bis}, \eqref{e4_bis} }}{\leq} C(n,\delta,\eta) \frac 23 \sup_{x\in\mathbb R^{n+1}}\left( \frac{n|a^2-1|}{N^2} |x|^2 \exp\left( n \frac{\alpha +1}{N} |x|-|x| \right)\right) \nonumber\\
&\times \sup_{x\in\mathbb R^{n+1}}\left( \sum_{k=0}^{N-1}\exp\left(k n\left( \frac{\alpha+1}{N} \right) |x| +an \frac{N-1-k}{N}|x|-2n(\alpha+1)|x|\right)\right) \nonumber\\
& \leq C(n,\delta,\eta) \frac 23 \sup_{x\in\mathbb R^{n+1}}\left( \frac{n|a^2-1|}{N} |x|^2 \exp\left( n \frac{\alpha +1}{N} |x|-|x| \right)\right)\nonumber
\end{eqnarray}
where the last estimate is due to the fact that
$$
\sup_{x\in\mathbb R^{n+1}}\left( \sum_{k=0}^{N-1}\exp\left(k n\left( \frac{\alpha+1}{N} \right) |x| +an \frac{N-1-k}{N}|x|-2n(\alpha+1)|x|\right)\right) \leq N
$$
for $N$ big enough. The estimate \eqref{estimate_cont_mono} implies that $F_N(x,a)\to E(xa)$ in $A_1$ when $N\to+\infty$.
\end{proof}

\subsection{The supershift in one Clifford variable}
Similarly, as we did in the slice hyperholomorphic setting, we now study the monogenic case using what we have proved above.

\begin{definition}\label{d5}
Let $|a|>1$. Let
 $\{ h_{j}(N) \}_{j\leq N,N\in\mathbb N_0}$ be a real-valued sequence such that
$$
\sup_{j \in\mathbb N,\ N\in\mathbb{N}} \  |h_{j}(N)|\leq 1.
$$
Let $G(\lambda)$, be an entire left monogenic function.
We say that
the sequence
\begin{equation*}
F_N(x,a)=\sum_{j=0}^N Z_j(N,a)
G( x h_{j}(N)),
\end{equation*}
where  $\{ Z_j(N,a)\}_{j,N}$, $j=0,\ldots ,N$, for $ N\in \mathbb{N}$ is a  Clifford-valued sequence,
admits the supershift property if
$$
\lim_{N\to \infty}F_N(x,a)=
G(x a),
$$
uniformly on any compact subset of $\mathbb R^{n+1}$.
\end{definition}

\begin{theorem}\label{bounded_operatorq4_bis}
Let $G$ be an entire left monogenic function whose series expansion at zero is given by
\begin{equation*}
G(x)=\sum_{m=0}^\infty P_{m}(x) G_{m}.
\end{equation*}
then the operators
$$
  \mathcal V(x,\partial_{\uny})f(y):=\sum_{s=0}^\infty \frac 1{c_s} \sum_{|t|=s} \frac 1{t!} \partial_{\uny}^t \left(f(y)\right) \sum_{|m|=s} \sum_{i+j=m} \frac{m!}{i!j!} x_0^{|i|} e_1^{i_1}\cdots e_n^{i_n} (x_1)^{j_1}\cdots (x_n)^{j_n} G_m
$$
is a family of continuous operators from $A_1$ to $A_1$ depending on the parameters $x$ where $$\partial_{\uny}^{m}=\partial_{y_1}^{m_1}\cdots \partial_{y_n}^{m_n},$$ $c_s=(-n)^s$ if $s$ is even or $c_s=(-n)^s(e_1+\dots+e_n)$ if $m$ is odd. If $x$ belongs to a compact subset $K\subseteq\mathbb R^{n+1}$, then the family is uniformly continuous with respect to $x\in K$.
\end{theorem}
\begin{proof}
Let $f\in A_1$. We consider its Taylor series:
$$f(q)=\sum_{\mathbf{\ell}\in\mathbb N^n_0} P_{\mathbf{\ell}}(x) a_\ell .
$$ In particular, by Lemma \ref{taylor_coeff_estimate} we have that there exist positive constants $C_f$ and $b_f$ such that
$$|a_{\mathbf{\ell}}|\leq C_f\frac{b^{|{\mathbf{\ell}}|}_f c(n,{\mathbf{\ell}})}{ |{\mathbf{\ell}}|!}.$$
Moreover, if $f_m\to 0$ in $A_1$ then $C_{f_m}\to 0$ and $b_{f_m}$ is bounded.
We have
 \begin{align}
& \left| \mathcal V(x,\partial_{\uny})(f(y)) \right|  =\left| \sum_{s=0}^\infty \frac 1{c_s} \sum_{|t|=s} \frac 1{t!} \partial_{\uny}^t \left(f(y)\right) \sum_{|m|=s} \sum_{i+j=m} \frac{m!}{i!j!} x_0^{|i|} e_1^{i_1}\cdots e_n^{i_n} (x_1)^{j_1}\cdots (x_n)^{j_n} G_m \right| \nonumber \\
 &= \left| \sum_{s=0}^\infty \frac 1{c_s} \sum_{|t|=s} \frac 1{t!} \partial_{\uny}^t \left( \sum_{\mathbf{\ell}\in\mathbb N^n_0} P_{\mathbf{\ell}}(y) a_\ell \right) \sum_{|m|=s} \sum_{i+j=m} \frac{m!}{i!j!} x_0^{|i|} e_1^{i_1}\cdots e_n^{i_n} (x_1)^{j_1}\cdots (x_n)^{j_n} G_m \right| \nonumber \\
 &\leq  2^{n/2}  \sum_{s=0}^\infty \frac 1{|c_s|} \sum_{|t|=s} \frac {1}{t!} \left( \sum_{\mathbf{\ell}\geq t } \frac{\ell!}{(\ell-t)!} \left| P_{\ell- t}(y) \right| \left| a_\ell \right| \right) \sum_{|m|=s} \sum_{i+j=m} \frac{m!}{i!j!} |x_0^{|i|}|  |x_1|^{j_1}\cdots |x_n|^{j_n} |G_m| \label{est_cont_v}\\
 &= 2^{n/2} \sum_{s=0}^\infty \frac 1{|c_s|} \sum_{|t|=s} \left( \sum_{\mathbf{\ell}\in\mathbb N^n_0 } \frac 1{t!} \frac{(\ell+t)!}{\ell!} \left| P_{\ell}(y) \right| \left| a_{\ell+t} \right| \right) \sum_{|m|=s} \sum_{i+j=m} \frac{m!}{i!j!} |x_0^{|i|}| |x_1|^{j_1}\cdots |x_n|^{j_n} |G_m| \nonumber \\
 &\leq 2^{n/2} C_f \sum_{s=0}^\infty \frac 1{|c_s|} \sum_{|t|=s} \left( \sum_{\mathbf{\ell}\in\mathbb N^n_0 } \frac{(\ell+t)!}{\ell!} \left| P_{\ell}(y) \right| \frac{c(n,\ell+t) b_f^{|\ell+t|}}{|\ell+t|!} \right) \nonumber \\
 &\times \sum_{|m|=s} \sum_{i+j=m} \frac{m!}{i!j!} |x_0^{|i|}| |x_1|^{j_1}\cdots |x_n|^{j_n} |G_m| . \nonumber
\end{align}
Now we prove the estimate
$$c(n,\ell+t)\leq (n-1)! \, 2^{2n+|\ell|+|t|-2}c(n,\ell)c(n,t).
$$
 Indeed, by the estimate
 $$(j+k)!\leq 2^{|j|+|k|}j!\, k!$$  for $j=n+|\ell|-1$ and $k=n+|t|-1$ we have
\[
\begin{split}
c(n,\ell+t)= \frac{(n+|\ell| + |t| -1)!}{(n-1)!(\ell + t)!} & \leq (n-1)! 2^{2n+|t|+|\ell| -2} \frac{(n+|\ell| -1)!}{(n-1)! \ell !} \frac{(n+|t|-1)!}{(n-1)! t!}\\
& \leq (n-1)! 2^{2n+|t|+|\ell| -2} c(n,\ell) c(n, t).
\end{split}
\]
Plugging this estimate in \eqref{est_cont_v} and observing that $\frac{(\ell+t)!}{|\ell+t|!}\leq 1$, we have that
\begin{equation}\label{est_cont_v_2}
\begin{split}
& \left| \mathcal V(x,\partial_{\uny})(f(y)) \right|  \leq C_f(n-1)! 2^{2n-2} \left( \sum_{\mathbf{\ell}\in\mathbb N^n_0 } \frac{c(n,\ell) (2b_f)^{|\ell|}}{\ell!} \left |P_\ell (y)\right |  \right)\\
&  \times\sum_{s=0}^\infty \frac 1{|c_s|} \left( \sum_{|t|=s} \frac{c(n,t)}{t!} \right) \sum_{|m|=s}  \sum_{i+j=m} \frac{m!}{i!j!} |x_0^{|i|}| \frac{|e_i|}{|i|!} |x_1|^{j_1}\cdots |x_n|^{j_n} c(n,m) (2b_f)^{s} |G_m|\\
&\leq C_f(n-1)! 2^{2n-2} \left( \sum_{\mathbf{\ell}\in\mathbb N^n_0 } \frac{c(n,\ell) (2b_f)^{|\ell|}}{\ell!} \left |P_\ell (y)\right |  \right)  \sum_{s=0}^\infty \frac 1{|c_s|} \left( \sum_{|t|=s} \frac{c(n,t)}{t!} \right) \\
&\times \sum_{s=0}^\infty \sum_{|m|=s}  \sum_{i+j=m} \frac{m!}{i!j!} |x_0^{|i|}| \frac{|e_i|}{|i|!} |x_1|^{j_1}\cdots |x_n|^{j_n} c(n,m) (2b_f)^{s} |G_m|.\\
\end{split}
\end{equation}
Now we observe that $|P_\ell(y)| \leq P_\ell(y')$ for $y=y_0+y_1e_1+\dots+ y_ne_n$ and $$y'=(\sqrt{y_0^2+y_1^2})e_1+\dots +(\sqrt{y_0^2+y_n^2}) e_n.$$
 This is true because, if we define $z_i=y_i- y_0e_i$
$$|P_\ell(y)|\leq |z_1|^{\ell_1} \cdots |z_n|^{\ell_n}=P_\ell(y').$$
By the fact that
 $$\limsup_{p\to +\infty} \left(\sum_{|m|=p} c(n,m)\right)^{\frac 1p}=n$$
(see Lemma 1 \cite{CAK07}) we can deduce that for any $\epsilon>0$ there exists $p_\epsilon>0$ such that for any $p\geq p_\epsilon$ we have
$$\left(\sum_{|m|=p} c(n,m)\right)^{\frac 1p}\leq n+\epsilon.$$ Let $C_\epsilon>1$ be such that $\sum_{|m|=p} c(n,m) \leq C_\epsilon (n+\epsilon)^p$ for any $0\leq p\leq p_\epsilon$. Thus we can conclude that for any $p>0$ and any multi-index $m\in\mathbb N^n$ with $|m|=p$ we have:
$$c(n,m)\leq\sum_{|m|=p} c(n,m)\leq C_\epsilon(n+\epsilon)^p.$$
Moreover, $|c_s|\geq n^{s}$ because when $s$ is odd $|e_1+\dots +e_n|=\sqrt{n}$. Plugging these two estimates in \eqref{est_cont_v_2} we obtain
\begin{equation}\label{est_cont_v_3}
\begin{split}
& \left| \mathcal V(x,\partial_{\uny})(f(y)) \right|  \leq C_\epsilon C_f(n-1)! 2^{2n-2} \left( \sum_{\mathbf{\ell}\in\mathbb N^n_0 } \frac{c(n,\ell) (2b_f)^{|\ell|}}{\ell!} P_\ell(y')  \right) \left( \sum_{t\in\mathbb N^n_0} \frac {c(n,t)}{t! n^{|\ell|}}  \right)  \\
&\times\sum_{s=0}^\infty \frac 1{|c_s|} \sum_{|m|=s}  \sum_{i+j=m} \frac{m!}{i!j!} |\frac {(n+\epsilon)}{n}b_f 2 x_0|^{|i|}  | \frac {(n+\epsilon)}{n}b_f 2x_1|^{j_1}\cdots | \frac {(n+\epsilon)}{n}b_f 2x_n|^{j_n}  |G_m|.
\end{split}
\end{equation}
By Lemma \ref{growth_taylor_series} we have
$$\sum_{\ell\in\mathbb N^n_0} \frac{c(n,\ell)}{\ell!} 2^{|\ell|}P_\ell(b_f y')=\sum_{\ell\in\mathbb N^n_0} \frac{c(n,\ell)}{\ell!} P_\ell(2b_f y')\leq e^{2 b_f |y'|} \leq e^{2nb_f |x|},$$
and
 $$\sum_{t\in\mathbb N^n_0} \frac {c(n,t)}{t! n^{|t|}} \leq e^{\frac 1n}$$ so the property (V) in Theorem \ref{fpt} we have that
\[
\begin{split}
& \sum_{i+j=m} \frac{m!}{i!j!} |\frac {(n+\epsilon)}{n}b_f 2 x_0|^{|i|}  | \frac {(n+\epsilon)}{n}b_f 2x_1|^{j_1}\cdots | \frac {(n+\epsilon)}{n}b_f 2x_n|^{j_n}\\
&= P_m \left ( \frac {(n+\epsilon)}{n}2 b_f ((|x_0|+|x_1| )e_1+(|x_0|+ |x_2|) e_2+\dots +(|x_0|+|x_n|)e_n)\right).
\end{split}
\]
Thus the estimate \eqref{est_cont_v_3} becomes
\begin{equation}\label{est_cont_v_4}
\begin{split}
& \left| \mathcal V(x,\partial_{\uny})(f(y)) \right|  \leq C_\epsilon C_f(n-1)! 2^{n-2} e^{2 n b_f |x| +\frac 1n}\\
&\times \sum_{s=0}^\infty  \sum_{|m|=s} P_m \left (\frac{ (n+\epsilon)}{n}2 b_f ((|x_0|+|x_1| )e_1+\dots +(|x_0|+|x_n|)e_n) \right) |G_m|.
\end{split}
\end{equation}
The previous estimate implies that $(\mathcal V( x, \partial_{\uny}) f(y))\to 0$ in $A_1$ if $f_m\to 0$ in $A_1$ because the series
 $$
 \sum_{s=0}^\infty  \sum_{|m|=s} P_m \left (\frac{ (n+\epsilon)}{n}2 b_f ((|x_0|+|x_1| )e_1+(|x_0|+ |x_2|) e_2+\dots +(|x_0|+|x_n|)e_n) \right) |G_m|
 $$
 is convergent due to the fact that the Taylors series of $G$ is absolutely convergent. Thus the uniformly continuity with respect to $x\in K$ of the family of operators $\mathcal V(x,\partial_{\uny})$ is proved because of the continuity in $x$ of the right hand side of \eqref{est_cont_v_4}.
\end{proof}

\begin{theorem}\label{propm=4q_bis}
Let $|a|>1$ and let
\begin{equation} \nonumber
f_N(x_1,\dots, x_n,a):= \sum_{j=0}^N e^{h_j(N)x} Z_j(N,a),\ \ \ N\in \mathbb{N},\ \ \ (x_1,\dots ,x_n)\in \mathbb{R}^{n},
\end{equation}
be a superoscillating function as in Definition \ref{SUPOSONE} and
assume that the sequence
$$f_N(x,a):=\sum_{j=0}^N E(h_j(N) x) Z_j(N,a)$$
converges to $E(ax)$ in the space $A_{1}$.
Let $G$ be entire left monogenic functions. We define
$$
F_N(x_0, x_1,\dots ,x_n,a)=\sum_{j=0}^N Z_j(N,a)G(h_j(N)x),
$$
Then, $F_N(x_0,x_1,\dots,\,x_n,a)$ admits the supershift property that is
$$
\lim_{N\to \infty}F_N(x_0, x_1,\dots, x_n,a)=G(x a)
$$
uniformly on any compact subsets of $\mathbb R^{n+1}$.
\end{theorem}
\begin{proof}
Let
\begin{equation*}
\sum_{|m |=0}^\infty P_{m} (\lambda) G_{m} .
\end{equation*}
be the Taylor series of $G(\lambda)$. First we observe that
\[
\begin{split}
&F_N( x,a)=\sum_{j=0}^NZ_j(N,a)G(x h_j(N))\\
&= \sum_{j=0}^NZ_j (N,a)\sum_{|m|=0}^\infty P_m(x h_j(N)) G_m=\sum_{j=0}^NZ_j (N,a)\sum_{|m|=0}^\infty P_m(x_0 h_j(N)+\unx h_j(N)) G_m\\
&=\sum_{j=0}^NZ_j (N,a)\sum_{s=0}^\infty\sum_{|m|=s} \sum_{i+j=m} \frac{m!}{i!j!} P_i(x_0 h_j(N))P_j(\unx h_j(N)) G_m\\
&=\sum_{j=0}^NZ_j (N,a)\sum_{s=0}^\infty\sum_{|m|=s} \sum_{i+j=m} \frac{m!}{i!j!} x_0^{|i|} e_1^{i_1}\cdot\dots \cdot e_n^{i_n}h_j(N)^{|i|} (x_1)^{j_1}\cdots (x_n)^{j_n} h_j(N))^{|j|} G_m
\end{split}
\]
\[
\begin{split}
&=\sum_{s=0}^\infty\sum_{|m|=s} \sum_{j=0}^N (h_j(N))^s Z_j (N,a)  \sum_{i+j=m} \frac{m!}{i!j!} x_0^{|i|} e_1^{i_1}\cdot\dots \cdot e_n^{i_n} (x_1)^{j_1}\cdots (x_n)^{j_n} G_m\\
&=\sum_{s=0}^\infty \sum_{|t|=s} \frac{1}{c_s t!} \partial_{\uny}^t \left( \sum_{j=0}^N E(h_j(N)y) Z_j (N,a) \right)\Big|_{y=0}
\\
&\times \sum_{|m|=s} \sum_{i+j=m} \frac{m!}{i!j!} x_0^{|i|}  e_1^{i_1}\cdot\dots \cdot e_n^{i_n} (x_1)^{j_1}\cdots (x_n)^{j_n} G_m
\end{split}
\]
where in the last equation we used Proposition \ref{derivative_mon_exp}. We define the operator
$$
 (\mathcal V(x,\partial_{\uny})f(y)):=\sum_{s=0}^\infty\sum_{|t|=s} \partial_{\uny}^t \left(f(y)\right) \sum_{|m|=s} \sum_{i+j=s} \frac{m!}{i!j!} x_0^{|i|} \frac{e_i}{|i|!} (x_1)^{j_1}\cdots (x_n)^{j_n} g_m
$$
and, by Theorem \ref{bounded_operatorq4_bis} we have that
\[
\begin{split}
& \lim_{N\to \infty} \sum_{j=0}^NZ_j(N,a)G(x h_j(N))=\lim_{N\to \infty} \sum_{j=0}^N  \mathcal V(x,\partial_{\uny})(E(h_j(N)y) Z_j(N,a) )\Big |_{y=0}\\
&=\mathcal V(x,\partial_{\uny}) \left(\lim_{N\to \infty} \sum_{j=0}^N  E(h_j(N)y) Z_j(N,a) \right) \Big |_{y=0}=\mathcal V(x,\partial_{\uny}) \left( E(ay) \right ) \Big |_{y=0}\\
&= \sum_{s=0}^\infty  a^s \left( (x)^s g_s \right)=G(x a),
\end{split}
\]
uniformly on $x\in K$ for any compact subset $K\subset \mathbb R^{n+1}$.
\end{proof}

\subsection{Superoscillating and supershift sequences in several Clifford variables.}

A similar remark has to be made here as we did for several variables in the slice hyperholomorphic setting. Here, we mean that we encode different frequencies in the Clifford variables, and we will investigate what happens with this natural modification.

\begin{definition}[Superoscillating sequence in several Clifford variables]\label{superoscill2}
A sequence of the type
\begin{equation}
F_N(x_0,\ldots ,x_n,a)=\sum_{j=0}^N Z_j(N,a)  E(e_0 x_0 h_{j,0}(N) +e_1x_1 h_{j,1}(N)+\dots+e_n x_n h_{j,n}(N)),
\end{equation}
where  $\{Z_j(N,a)\}_{j,N}$, for $j=0,\ldots ,N$ and  $ N\in \mathbb{N}$, is a  Cliffordian-valued sequence, is said to be {\em a superoscillating sequence} if
$$
\sup_{j=0,\ldots ,N} \  |h_{j,\ell}(N)|\leq 1 ,\ \ {\rm for} \ \ell=0,...,n,
$$
 and $F_N(x_0,\ldots ,x_n,a)$ converges uniformly on $x$ in any compact subset of $\mathbb R^{n+1}$ to
 $$E(e_0 x_0 g_0(a) +e_1x_1 g_1(a)+\dots +e_nx_n g_n(a)),$$ where $a\in\mathbb R$, $g_\ell$'s are continuous functions of a real variable whose domain is $\mathbb R$ and $|g_\ell (a)|>1$ for  $\ell=1,\ldots ,d$.
\end{definition}

\begin{theorem}\label{bounded_operator_2}
Let $u_{\ell}:\mathbb R^{n+1} \to\mathbb R^{n+1}$ be a sequence of continuous functions for $\ell\in\mathbb N$ such that for any compact subset $K\subset \mathbb R^{n+1}$
$$ \lim \sup_{\ell\to\infty} \sup_{x\in K} |u_{\ell} (x) |^{1/\ell} = 0\quad\textrm{and}\quad u_{\ell}(x)= u_{\ell,0}(x) +u_{\ell,1} (x) e_1+\dots+u_{\ell,n}(x) e_n,$$
where $u_{\ell,i}:\mathbb R^{n+1}\to\mathbb R$ are $C^0$ functions. Then the operator
\[
\begin{split}
&\left( \mathcal U(x,\partial_{\uny}) f(y) \right) :=  \sum_{k=0}^\infty \frac{1}{k!}  \sum_{|t|=k} \sum_{i+r=t} \frac{t!}{i!} \sum_{s=0}^\infty \frac 1{c_s} \sum_{z'\in\mathbb N^n,\, |z'|=s}^\infty \frac 1{z'!} \partial_{\uny}^{z'}\left( f(y) \right) \frac {e_i}{|i|!} \\
&\times  \sum_{z\in\mathbb N^{n+1},\, |z|=s} \left(\sum_{s_1+\dots+s_{|i|}=z_0} u_{s_1,0}(x)\cdots u_{s_{|i|},0}(x)\right) \frac{1}{r_1!} \left(\sum_{s_1+\dots+s_{r_1}=z_1} u_{s_1,1}(x) \cdots u_{s_{r_1},1}(x) \right) \cdots\\
&\times \frac 1{r_n!} \left(\sum_{s_1+\dots+s_{r_n}=z_n} u_{s_1,n}(x)\cdots  u_{s_{r_n},n} (x) \right ) e_t'
\end{split}
\]
is a continuous operator from $A_1$ to $A_1$ where $$\partial_{\uny}^{\mathbf{z'}}=\partial_{y_1}^{z_0+z_1}\cdots \partial_{y_n}^{z_n},$$ $c_s=(-n)^s$ if $s$ is even or $c_s=(-n)^s(e_1+\dots+e_n)$ if $m$ is odd, $e_i$ is defined for any multi-index $i\in \mathbb N^n$ as in \eqref{sum_ima_units_2} and $e_t'$ is defined as in \eqref{sum_ima_units} for any multi-index $i\in \mathbb N^n$.
\end{theorem}
\begin{proof}
Let $f\in A_1$. We consider its Taylor series: $f(q)=\sum_{\mathbf{\ell}\in\mathbb N^n_0} P_{\mathbf{\ell}}(x) a_\ell $. In particular, by Lemma 2.13 in \cite{CPSS21} we have that there exist constants $C_f$ and $b_f$ such that
$$|a_{\mathbf{\ell}}|\leq C_f\frac{b^{|{\mathbf{\ell}}|}_f c(n,{\mathbf{\ell}})}{ |{\mathbf{\ell}}|!},$$
where 
$$
c(n,\ell)=\frac{n(n+1)\cdots (n+|\ell|-1)}{\ell!}=\frac{(n+|\ell|-1)!}{(n-1)!\, \ell ! }.
$$
Moreover, if $f_m\to 0$ in $A_1$ then $C_{f_m}\to 0$ and $b_{f_m}$ is bounded.
We have
 \begin{eqnarray}\nonumber
 & \left| \mathcal U(x,\partial_{\uny}) f(y) \right|  =\left| \sum_{k=0}^\infty \frac{1}{k!}  \sum_{|t|=k} \sum_{i+r=t} \frac{t!}{i!} \sum_{s=0}^\infty \frac 1{c_s} \sum_{z'\in\mathbb N^n,\, |z'|=s}^\infty \frac 1{z'!} \partial_{\uny}^{z'}\left( f(y) \right) \frac {e_i}{|i|!} \right. \nonumber\\
&\times  \sum_{z\in\mathbb N^{n+1},\, |z|=s} \left(\sum_{s_1+\dots+s_{|i|}=z_0} u_{s_1,0}(x)\cdots u_{s_{|i|},0}(x)\right)  \nonumber
\\
&
\times
\frac{1}{r_1!} \left(\sum_{s_1+\dots+s_{r_1}=z_1} u_{s_1,1}(x) \cdots u_{s_{r_1},1}(x) \right) \cdots \nonumber
\\
&\left. \times \frac 1{r_n!} \left(\sum_{s_1+\dots+s_{r_n}=z_n} u_{s_1,n}(x)\cdots  u_{s_{r_n},n} (x) \right ) e'_t\right| \nonumber
\\
 & = \left| \sum_{k=0}^\infty \frac{1}{k!}  \sum_{|t|=k} \sum_{i+r=t} \frac{t!}{i!} \sum_{s=0}^\infty \frac 1{c_s} \sum_{z'\in\mathbb N^n,\, |z'|=s}^\infty \frac 1{z'!} \partial_{\uny}^{z'}\left( \sum_{\ell\in\mathbb N^n_0} P_\ell(y) a_\ell \right) \frac {e_i}{|i|!} \right. \nonumber\\
&\times  \sum_{z\in\mathbb N^{n+1},\, |z|=s} \left(\sum_{s_1+\dots+s_{|i|}=z_0} u_{s_1,0}(x)\cdots u_{s_{|i|},0}(x)\right)  \nonumber
\\
&
\times
\frac{1}{r_1!} \left(\sum_{s_1+\dots+s_{r_1}=z_1} u_{s_1,1}(x) \cdots u_{s_{r_1},1}(x) \right) \cdots \nonumber\\
&\left. \times \frac 1{r_n!} \left(\sum_{s_1+\dots+s_{r_n}=z_n} u_{s_1,n}(x)\cdots  u_{s_{r_n},n} (x) \right ) e'_t\right| \nonumber\\
 & \leq 2^{n/2} \sum_{k=0}^\infty \frac{1}{k!}  \sum_{|t|=k} \sum_{i+r=t} \frac{t!}{i!} \sum_{s=0}^\infty \frac 1{|c_s|} \sum_{z'\in\mathbb N^n,\, |z'|=s}^\infty \frac 1{z'!} \left( \sum_{\ell\geq z'} \frac{\ell!}{(\ell-z')!} | a_\ell| |P_{\mathbf{\ell}-z'}(y)|  \right) \nonumber\\
&\times \sum_{z\in\mathbb N^{n+1},\, |z|=s} \left(\sum_{s_1+\dots+s_{|i|}=z_0} |u_{s_1,0}(x)| \cdots |u_{s_{|i|},0}(x)|\right) \nonumber
\\
&
\times \frac{1}{r_1!} \left(\sum_{s_1+\dots+s_{r_1}=z_1} |u_{s_1,1}(x)| \cdots |u_{s_{r_1},1}(x)| \right) \cdots \nonumber\\
& \times \frac 1{r_n!} \left(\sum_{s_1+\dots+s_{r_n}=z_n} |u_{s_1,n}(x)| \cdots  |u_{s_{r_n},n}(x)| \right ) |e'_t| \nonumber\\
& = 2^{n/2} \sum_{k=0}^\infty \frac{1}{k!}  \sum_{|t|=k} \sum_{i+r=t} \frac{t!}{i!} \sum_{s=0}^\infty \frac 1{|c_s|} \sum_{z'\in\mathbb N^n,\, |z'|=s}^\infty \frac 1{z'!}  \left( \sum_{\ell \in\mathbb N^n_0} \frac{(\ell+z')!}{\ell!} | a_{\ell+z'}| |P_{\ell}(y)|  \right)\nonumber \\
&\times \sum_{z\in\mathbb N^{n+1},\, |z|=s} \left(\sum_{s_1+\dots+s_{|i|}=z_0} |u_{s_1,0}(x)| \cdots |u_{s_{|i|},0}(x)|\right) \nonumber
\\
&
\times \frac{1}{r_1!} \left(\sum_{s_1+\dots+s_{r_1}=z_1} |u_{s_1,1}(x)| \cdots |u_{s_{r_1},1}(x)| \right) \cdots \nonumber\\
& \times \frac 1{r_n!} \left(\sum_{s_1+\dots+s_{r_n}=z_n} |u_{s_1,n}(x)| \cdots  |u_{s_{r_n},n}(x)| \right ) |e'_t|. \nonumber
\end{eqnarray}
where in the first inequality we also use that $\frac{|e_i|}{|i|!}\leq 1$. Thus we have
\begin{align*}
 \left| \mathcal U(x,\partial_{\uny}) f(y) \right|& \leq 2^{n/2} \sum_{k=0}^\infty \frac{1}{k!}  \sum_{|t|=k} \sum_{i+r=t} \frac{t!}{i!} \sum_{s=0}^\infty \frac 1{|c_s|}
\\
&\times \sum_{z'\in\mathbb N^n,\, |z'|=s}^\infty \frac 1{z'!}  \left( \sum_{\ell \in\mathbb N^n_0} \frac{ (\ell+z')!}{\ell!} \frac{b_f^{|\ell+z'|} c(n,\ell+z')}{|\ell+z'|! }  |P_\ell(y)|  \right)
\nonumber\\
&\times \sum_{z\in\mathbb N^{n+1},\, |z|=s} \left(\sum_{s_1+\dots+s_{|i|}=z_0} |u_{s_1,0}(x)| \cdots |u_{s_{|i|},0}(x)|\right)
\\
&\times \frac{1}{r_1!} \left(\sum_{s_1+\dots+s_{r_1}=z_1} |u_{s_1,1}(x)| \cdots |u_{s_{r_1},1}(x)| \right) \cdots\nonumber\\
 & \times \frac 1{r_n!} \left(\sum_{s_1+\dots+s_{r_n}=z_n} |u_{s_1,n}(x)| \cdots  |u_{s_{r_n},n}(x)| \right ) |e'_t|,
 \end{align*}
 since 
 $$c(n,\ell+z')\leq (n-1)! \, 2^{n+|\ell|+|z'|-2}c(n,\ell)c(n,z'),$$ as we proved in Theorem \ref{bounded_operatorq4_bis}, we have
 \begin{align*}
  \left| \mathcal U(x,\partial_{\uny}) f(y) \right|& \leq C_f (n-1)! 2^{(3n-4)/2}  \sum_{\ell \in\mathbb N^n_0} 2^{|\ell|} \frac{b_f^{|\ell|} c(n,\ell)}{\ell! }  |P_\ell(y)| \\
 & \times \sum_{k=0}^\infty \frac{1}{k!}  \sum_{|t|=k} \sum_{i+r=t} \frac{t!}{i!} 
 \sum_{s=0}^\infty \frac {b_f^s2^s}{|c_s|} \sum_{z'\in\mathbb N^n_0, |z'|=s} \frac{c(n,z')}{z'!}\\
 &\times \sum_{s=0}^\infty \sum_{z\in\mathbb N^{n+1}_0, |z|=s} \left(\sum_{s_1+\dots+s_{|i|}=z_0} |u_{s_1,0}(x)| \cdots |u_{s_{|i|},0}(x)|\right)
  \\
&
\times\frac{1}{r_1!} \left(\sum_{s_1+\dots+s_{r_1}=z_1} |u_{s_1,1}(x)| \cdots |u_{s_{r_1},1}(x)| \right) \cdots\\
 &\times \frac 1{r_n!} \left(\sum_{s_1+\dots+s_{r_n}=z_n} |u_{s_1,n}(x)| \cdots  |u_{s_{r_n},n}(x)| \right ) |e'_t|.
 \end{align*}
 Since $|P(y)| \leq P(y')$ for $y'=(\sqrt{y_0^2+y_1^2})e_1+\dots +(\sqrt{y_0^2+y_n^2}) e_n$, as we proved in Theorem \ref{bounded_operatorq4_bis}, and
 $$\limsup_{p\to +\infty}\left(\sum_{|z'|=p} c(n,z')\right)^{\frac 1p}=n$$
(see Lemma 1 \cite{CAK07}) which in particular implies that for any $\epsilon>0$ there exists a positive constant $C_\epsilon>0$ such that for any $p>0$: $\sum_{|z'|=p} c(n,z')\leq C_\epsilon(n+\epsilon)^p$ and $|c_s|\geq n^s$ (see the proof of Theorem \ref{bounded_operatorq4_bis}), we have

 \begin{align*}\nonumber
  & \left| \mathcal U(x,\partial_{\uny}) f(y) \right| \leq C_\epsilon C_f (n-1)! 2^{(3n-4)/2}  \sum_{\ell \in\mathbb N^n_0} 2^{|\ell|} \frac{b_f^{|\ell|} c(n,\ell)}{\ell! } P_\ell(y') \nonumber\\
 & \times  \sum_{k=0}^\infty \frac{1}{k!}  \sum_{|t|=k} \sum_{i+r=t} \frac{t!}{i! r!} \sum_{s=0}^\infty \frac {b_f^s2^s}{|c_s|} \sum_{z'\in\mathbb N^n_0, |z'|=s} \frac{c(n,z')}{z'!} \nonumber\\
 & \times \sum_{s=0}^\infty \sum_{z\in\mathbb N^{n+1}_0, |z|=s} \left(\sum_{s_1+\dots+s_{|i|}=z_0} |u_{s_1,0}(x)| \cdots |u_{s_{|i|},0}(x)|\right)\nonumber\\
 & \times  \left(\sum_{s_1+\dots+s_{r_1}=z_1} |u_{s_1,1}(x)| \cdots |u_{s_{r_1},1}(x)| \right) \cdots \left(\sum_{s_1+\dots+s_{r_n}=z_n} |u_{s_1,n}(x)| \cdots  | u_{s_{r_n},n}(x)| \right ) |e'_t|\nonumber
  \\
 &\leq C_\epsilon C_f (n-1)! 2^{(3n-4)/2}  \sum_{\ell \in\mathbb N^n_0} 2^{|\ell|} \frac{b_f^{|\ell|} c(n,\ell)}{\ell! } P_\ell(y') \sum_{s=0}^\infty \frac {b_f^s2^s}{|c_s|} \sum_{z'\in\mathbb N^n_0, |z'|=s} \frac{c(n,z')}{z'!} \nonumber
 \\
 & \times \sum_{k=0}^\infty \frac{1}{k!}  \sum_{|t|=k} \sum_{i+r=t} \frac{t!}{i! r!} P_i \left( \sum_{v=0}^n e_v\left( \sum_{\ell=0}^\infty |u_{\ell,0}(x)| \right) \right) P_r\left(  \sum_{v=1}^n \left( \sum_{\ell=0}^\infty  |u_{\ell,v}(x)| \right)e_v \right)|e_t'|\nonumber
 \end{align*}
   \begin{align*}
 &\leq C_\epsilon C_f (n-1)! 2^{(3n-4)/2}  \sum_{\ell \in\mathbb N^n_0} 2^{|\ell|} \frac{b_f^{|\ell|} c(n,\ell)}{\ell! } P_\ell(y') \sum_{s=0}^\infty \frac {b_f^s2^s}{|c_s|} \sum_{z'\in\mathbb N^n_0, |z'|=s} \frac{c(n,z')}{z'!}\nonumber
 \\
 & \times \sum_{k=0}^\infty \frac{1}{k!}  \sum_{|t|=k} P_{t} \left( \sum_{v=1}^n e_v \left( \sum_{\ell=0}^\infty |u_{\ell,0}(x)|  + |u_{\ell,v} (x)| \right)  \right) |e'_t|.\nonumber
 \end{align*}
By Lemma \ref{growth_taylor_series} we have
$$\sum_{s_1=0}^\infty \sum_{|\ell|=s_1} \frac{c(n,\ell)}{\ell!} P_\ell(2b_f y')\leq e^{2 b_f (|y'|+1/n)}$$ and
$$\sum_{s=0}^\infty \frac {b_f^s2^s}{|c_s|} \sum_{z'\in\mathbb N^n_0, |z'|=s} \frac{c(n,z')}{z'!}\leq e^{2b_f/n},
$$ thus we can obtain the following estimate
 \begin{equation}\label{continuity_estimate}
 \begin{split}
 & \left| \mathcal U(x,\partial_{\uny}) f(y) \right|  \leq C_f (n-1)! 2^{n-2}C_\epsilon e^{2 b_f (|y'|+1/n)}\\
 & \sum_{k=0}^\infty \frac{1}{k!}  \sum_{|t|=k} P_{t} \left( \sum_{v=1}^n e_v \left( \sum_{\ell=0}^\infty |u_{\ell,0}(x)|  + |u_{\ell,v} (x)| \right)  \right) |e'_t|\\
 & \overset{|y'|\leq \sqrt{n}|y|}{\leq}  C_f (n-1)! 2^{n-2}C_\epsilon e^{2 \sqrt{n} b_f (|y| + 1/n^{3/2})} \\
 & \times \sum_{k=0}^\infty \frac{1}{k!}  \sum_{|t|=k} P_{t} \left( \sum_{v=1}^n e_v \left( \sum_{\ell=0}^\infty |u_{\ell,0}(x)|  + |u_{\ell,v} (x)| \right)  \right) |e'_t|,
 \end{split}
\end{equation}
where the uniformly convergence on any compact subset of $\mathbb R^{n+1}$ of the last series is justified by Proposition \ref{exp_mon_funct}. Thus we proved $\mathcal U(x,\partial_{\uny})(f_m)\to 0$ in $A_1$ if $f_m\to 0$ in $A_1$  and the family of operators $\mathcal U(x,\partial_{\uny})$ is uniformly continuous with respect to $x\in K$ for any compact subset $K\subseteq \mathbb R^{n+1}$.
\end{proof}

\begin{theorem}\label{new_suposs}
Let $\left\{u_\ell: \mathbb R^{n+1}\to \mathbb R^{n+1}\right\}_{\ell\in\mathbb N}$ be a sequence of functions such that for any compact subset $K\subseteq\mathbb R^n$
$$
\lim \sup_{\ell\to\infty} \sup_{x\in K} \left( |u_{\ell} (x)|^{1/\ell}\right) = 0.
$$
Let $a\in\mathbb R$. Suppose also that the sequence
$$
 f_N (x_0,\dots ,x_n,a) :=  \sum_{j=0}^N E(h_j(N)  x) Z_j(N,a),
$$
converges to $ E(ax)$ in $A_{1}$ with in addition $h_j(N)\in[-1,1]$ for any $0\leq j\leq N$.
Let also
$$
F_N  (x_0,x_1,\ldots,x_n,a) := \sum_{j=0}^N Z_j(N,a) E\left(\sum_{\ell=0}^\infty u_\ell(x) h_j(N)^\ell \right).
$$
Then we have
$$
\lim_{N\to \infty}F_N(x_0,x_1,\ldots,x_n,a)=E\left(\sum_{\ell=0}^\infty u_\ell(x) a^\ell\right),
$$
uniformly on any compact subsets of $\mathbb R^{n+1}$.
\end{theorem}

\begin{proof}
We observe that $u_\ell(x)=u_{\ell,0}(x)+\sum_{v=1}^n e_v u_{\ell,v}(x)$ where $u_{\ell,0}$ is the real part of $u_\ell(x)$ and $u_{\ell,v}(x)$ are the real coefficients of the imaginary component of $e_v$ in $u_\ell(x)$. Let $t,\, r,\, i,\, z,\, z',\, z''$ be multi-indexes in $\mathbb N^n_0$ and $e'_t$ defined as in \eqref{sum_ima_units}. We have that
\begin{eqnarray}\label{opmon}
 & \sum_{j=0}^N Z_j(N,a) E\left(\sum_{\ell=0}^\infty u_\ell(x) h_j(N)^\ell\right)  \nonumber\\
 & = \sum_{j=0}^N Z_j(N,a) \sum_{k=0}^\infty \frac{1}{k!} \sum_{ |t|=k} P_t\left  (\sum_{\ell=0}^\infty\left( u_{\ell,0}(x)+\sum_{v=1}^n e_v u_{\ell,v}(x)\right) h_j(N)^\ell\right) e'_t\nonumber\\
& = \sum_{j=0}^N Z_j(N,a) \sum_{k=0}^\infty \frac{1}{k!} \sum_{|t|=k} \sum_{i+r=t} \frac{t!}{i!r!}P_i\left  (\sum_{\ell=0}^\infty  u_{\ell,0}(x) h_j(N)^\ell \right) \nonumber\\
&\quad\quad\quad\quad\quad\quad\quad\quad\quad\quad\quad\quad\quad\quad\quad\quad\quad\quad\quad\quad\times P_r\left( \sum_{\ell=0}^\infty \left(\sum_{v=1}^n e_v u_{\ell,v}(x) \right) h_j(N)^\ell\right) e'_t\nonumber\\
& \overset{\textrm{Prop. \ref{fpt}, Rem. \ref{remark_fp}}}{=} \sum_{j=0}^N Z_j(N,a) \sum_{k=0}^\infty \frac{1}{k!} \sum_{|t|=k} \sum_{i+r=t} \frac{t!}{i!} \left  (\sum_{\ell=0}^\infty  u_{\ell,0}(x) h_j(N)^\ell \right)^{|i|} \frac {e_i}{|i|!} \\
&\quad\quad\quad\quad\quad\quad \times \frac 1{r_1!} \left ( \sum_{\ell=0}^\infty u_{\ell,1}(x) h_j(N)^\ell\right)^{r_1}\cdots \frac 1{r_n!} \left ( \sum_{\ell=0}^\infty u_{\ell,n}(x) h_j(N)^\ell\right)^{r_n} e'_t\nonumber\\
& = \sum_{j=0}^N Z_j(N,a) \sum_{k=0}^\infty \frac{1}{k!} \sum_{|t|=k}  \nonumber
\\
&\times \sum_{i+r=t} \frac{t!}{i!} \left  ( \sum_{z_0=0}^\infty \sum_{s_1+\dots+s_{|i|}=z_0} u_{s_1,0}(x)\cdots u_{s_{|i|},0}(x) h_j(N)^{z_0} \right) \frac {e_i}{|i|!} \nonumber\\
&\quad\quad\quad\quad\quad\quad \times \frac 1{r_1!} \left( \sum_{z_1=0}^\infty \sum_{s_1+\dots+s_{r_1}=z_1} u_{s_1,1}(x) \cdots u_{s_{r_1},1}(x) h_j(N)^{z_1} \right)\cdots \nonumber\\
&\quad\quad\quad\quad\quad\quad \times \frac 1{r_n!} \left( \sum_{z_n=0}^\infty \sum_{s_1+\dots+s_{r_n}=z_n} u_{s_1,n}(x)\cdots  u_{s_{r_n},n}(x) h_j(N)^{z_n} \right) e'_t\nonumber\\
&=\sum_{j=0}^N Z_j(N,a) \sum_{k=0}^\infty \frac{1}{k!}  \sum_{|t|=k} \sum_{i+r=t} \frac{t!}{i!} \sum_{z_0=0}^\infty\dots \sum_{z_n=0}^\infty h_j(N)^{z_0+\dots+z_n} \frac {e_i}{|i|!} \nonumber\\
&\times \left(\sum_{s_1+\dots+s_{|i|}=z_0} u_{s_1,0}(x)\cdots u_{s_{|i|},0} (x)\right) \frac{1}{r_1!} \left(\sum_{s_1+\dots+s_{r_1}=z_1} u_{s_1,1}(x) \cdots u_{s_{r_1},1}(x) \right) \cdots \nonumber\\
&\times \frac 1{r_n!} \left(\sum_{s_1+\dots+s_{r_n}=z_n} u_{s_1,n}(x)\cdots  u_{s_{r_n},n}(x) \right ) e'_t.\nonumber
\end{eqnarray}
By Proposition \ref{derivative_mon_exp}, equation \eqref{opmon} becomes
\[
\begin{split}
&\sum_{j=0}^N Z_j(N,a) E\left(\sum_{\ell=0}^\infty u_\ell(x) h_j(N)^\ell\right)\\
&= \sum_{k=0}^\infty \frac{1}{k!}  \sum_{|t|=k} \sum_{i+r=t} \frac{t!}{i!} \sum_{s=0}^\infty \frac 1{c_s} \sum_{z'\in\mathbb N^n, |z'|=s} \frac 1{z'!} \partial_{\uny}^{z'}\left( \sum_{j=0}^N E(h_j(N)y) Z_j(N,a) \right)\Big|_{y=0} \frac {e_i}{|i|!}\\
&\times \sum_{z\in\mathbb N^{n+1},\, |z|=s}  \left(\sum_{s_1+\dots+s_{|i|}=z_0} u_{s_1,0}(x)\cdots u_{s_{|i|},0} (x)\right) \frac{1}{r_1!} \left(\sum_{s_1+\dots+s_{r_1}=z_1} u_{s_1,1}(x) \cdots u_{s_{r_1},1}(x) \right) \cdots\\
&\times \frac 1{r_n!} \left(\sum_{s_1+\dots+s_{r_n}=z_n} u_{s_1,n}(x)\cdots  u_{s_{r_n},n} (x) \right ) e'_t.
\end{split}
\]
We define the operator
\[
\begin{split}
&\left( \mathcal U(x,\partial_{\uny}) f(y) \right) :=  \sum_{k=0}^\infty \frac{1}{k!}  \sum_{|t|=k} \sum_{i+r=t} \frac{t!}{i!} \sum_{s=0}^\infty \frac 1{c_s} \sum_{z'\in\mathbb N^n,\, |z'|=s}^\infty \frac 1{z'!} \partial_{\uny}^{z'}\left( f(y) \right) \frac {e_i}{|i|!} \\
&\times \sum_{z\in\mathbb N^{n+1},\, |z|=s} \left(\sum_{s_1+\dots+s_{|i|}=z_0} u_{s_1,0}(x)\cdots u_{s_{|i|},0}(x)\right) \frac{1}{r_1!} \left(\sum_{s_1+\dots+s_{r_1}=z_1} u_{s_1,1}(x) \cdots u_{s_{r_1},1}(x) \right) \cdots\\
&\times \frac 1{r_n!} \left(\sum_{s_1+\dots+s_{r_n}=z_n} u_{s_1,n}(x)\cdots  u_{s_{r_n},n} (x) \right ) e'_t.
\end{split}
\]
and, since in the proof of Theorem \ref{bounded_operator_2} we proved the estimate \eqref{continuity_estimate}, the operators $\mathcal U(x,\partial_{\uny})$ is a family of uniformly  continuous operators from $A_1$ to $A_1$, depending on the parameters $x\in K\subset \mathbb R^{n+1}$ where $K$ is a compact subset. Thus
we have that
\[
\begin{split}
& \lim_{N\to \infty} \sum_{j=0}^N Z_j(N,a) E\left (\sum_{\ell=0}^\infty u_\ell(x) h_j(N)^\ell\right) =\lim_{N\to \infty} \sum_{j=0}^N  \mathcal U(x,\partial_{\uny}) \left( E\left(h_j(N) y\right) Z_j(n,a) \right)\big |_{y=0}\\
&=\mathcal U(x,\partial_{\uny}) \left(\lim_{N\to \infty} \sum_{j=0}^N  E\left(h_j(N) x\right) Z_j(N,a) \right) \big |_{ y=0}\\
&=\mathcal U(x,\partial_{\uny}) \left( E(a y) \right ) \big |_{ y=0}\\
&=  \sum_{k=0}^\infty \frac{1}{k!}  \sum_{|t|=k} \sum_{i+r=t} \frac{t!}{i!} \sum_{s=0}^\infty \sum_{z_0=0}^\infty \sum_{z'=(z_0+z_1,z_2,\dots, z_n)\in\mathbb N^n_0, |z'|=s}  a^s\\ &\times \frac{e_i}{|i|!} \left(\sum_{s_1+\dots+s_{|i|}=z_0} u_{s_1,0}(x)\cdots u_{s_{|i|},0}(x)\right) \frac{1}{r_1!} \left(\sum_{s_1+\dots+s_{r_1}=z_1} u_{s_1,1}(x) \cdots u_{s_{r_1},1}(x) \right) \cdots\\
&\times \frac 1{r_n!} \left(\sum_{s_1+\dots+s_{r_n}=z_n} u_{s_1,n}(x)\cdots  u_{s_{r_n},n} (x) \right ) e'_t\\
&=E\left(\sum_{\ell=0}^\infty u_\ell(x) a^\ell\right)
\end{split}
\]
where the last equality is due to the chain of equalities in \eqref{opmon}, read in inverse sense,  and the convergence of the last series to the exponential function is uniform on any compact subset $K\subseteq \mathbb R^{n+1}$.
\end{proof}

\begin{corollary}\label{cor2}
Let $g_k(a):=\sum_{\ell=0}^\infty g_{k,\ell}a^\ell$ for $k=0,\dots, n$ be real analytic functions such that $|g_k(a)|>1$ for $|a|>1$ and $|g_k(a)|\leq 1$ for $|a|\leq1$, if we suppose $u_\ell (x)=e_0x_0 g_{0,\ell}+\dots+e_n x_n g_{n,\ell}$ and the sequence $f_N(x,a)=\sum_{j=0}^N  E\left( h_j(N) x\right) Z_j(N,a) $ converges to $ x\mapsto E\left( ax\right)$ in $A_{1}$,  then the sequence
\[
\begin{split}
F_N(x,a)&=\sum_{j=0}^N Z_j(N,a) E\left (\sum_{\ell=0}^\infty u_\ell(x) (h_j(N))^\ell\right) \\
&=\sum_{j=0}^N Z_j(N,a) E\left (x_0 e_0 g_0(h_j(N))+x_1e_1 g_1(h_j(N))+\dots +x_n e_n g_n(h_j(N))\right)
\end{split}
\]
converges uniformly on compact subsets of $\mathbb R^{n+1}$ to $E\left( e_0 x_0g_0(a)+\dots+ e_n x_ng_n(a)\right)$ and thus it is a superoscillating sequence according to Definition \ref{superoscill2}.
\end{corollary}
\begin{proof}
The proof is a direct consequence of Theorem \ref{new_suposs}, once it is observed that, since the functions $g_k$'s are entire functions, we have
\[
\begin{split}
\limsup_{\ell\to\infty} \sup_{x\in K} |u_{\ell} (x)|^{1/\ell} & \leq  \limsup_{\ell\to\infty} \sup_{x\in K} \sum_{k=0}^n |x_k|^{1/\ell} |g_{k,\ell}|^{1/\ell} \\
& \leq  \limsup_{\ell\to\infty}  \sup_{x\in K} \left( \sum_{k=0}^n |x_k|^{1/\ell}\right ) \sum_{k=0}^n |g_{k,\ell}|^{1/\ell} =0
\end{split}
\]
for any compact subset $K\subset \mathbb R^{n+1}$.
\end{proof}

\begin{definition}[Supershift sequences in several Clifford variables]\label{ss2}
Let $|a|>1$. Let
 $\{ h_{j,\ell}(N) \}$,  $j=0,...,N$  for $ N\in \mathbb{N}$, be  real-valued sequences for $\ell=0,...,n$
 such that
$$
\sup_{j=0,\ldots ,N,\ N\in\mathbb{N}} \  |h_{j,\ell}(N)|\leq 1 ,\ \ {\rm for} \ \ell=0,...,n.
$$
Let $G(\lambda)$, be an entire left monogenic function.
We say that
the sequence
\begin{equation}\label{basic_sequence_sevFF_6}
F_N(x,a)=\sum_{j=0}^N Z_j(N,a)
G( x_0 h_{j,0}(N) +e_1x_1 h_{j,1}(N)+\dots +e_n x_n h_{j,n}(N)),
\end{equation}
where  $\{ Z_j(N,a)\}_{j,N}$, $j=0,\ldots ,N$, for $ N\in \mathbb{N}$ is a  Clifford-valued sequence,
admits the supershift property if
$$
\lim_{N\to \infty}F_N(x,a)=
G(x_0a +e_1x_1a+\dots +e_n x_n a),
$$
uniformly on any compact subset of $\mathbb R^{n+1}$.
\end{definition}

We now study the continuity of a suitable infinite-order differential operator that we will use in the sequel to examine the supershift property in the monogenic case.
\begin{theorem}\label{bounded_operatorq6}
Let $G$ be entire left monogenic function whose series expansions at zero is given by
whose series expansions at zero is given by
\begin{equation*}
G (\lambda)=\sum_{m\in\mathbb N^n_0} P_{m}(\lambda) G_{ m},
\end{equation*}
and $g_0,\dots, g_n$ be real analytic on $\mathbb R$ and real valued functions then the operators
\[
\begin{split}
& \mathcal V(x,\partial_{\uny})(f(y)):= \sum_{k=0}^\infty \frac{1}{k!}  \sum_{|t|=k} \sum_{i+r=t} \frac{t!}{i!} \sum_{s=0}^\infty \frac 1{c_s} \sum_{z'\in\mathbb N^n,\, |z'|=s} \frac 1{z'!} \partial_{\uny}^{z'}\left( f(y) \right) \frac{e_i}{|i|!} \\
&\times \sum_{z\in\mathbb N^{n+1},\, |z|=s} \left(\sum_{s_1+\dots+s_{|i|}=z_0} x_0 g_{s_1,0}\cdots x_0 g_{s_{|i|},0}\right) \frac{1}{r_1!} \left(\sum_{s_1+\dots+s_{r_1}=z_1} x_1 g_{s_1,1} \cdots x_1g_{s_{r_1},1} \right) \cdots\\
&\times \frac 1{r_n!} \left(\sum_{s_1+\dots+s_{r_n}=z_n} x_n g_{s_1,n}\cdots  x_n g_{s_{r_n},n} \right ) G_t,
\end{split}
\]
where $g_{j,\ell}$ are the Taylor coefficients of $g_i$ for $i=0,\dots, n$ and $\ell\in\mathbb N_0$, form a family of contiuous operators from $A_1$ to $A_1$ depending on the parameters $x$. If $x$ belongs to a compact subset $K\subseteq\mathbb R^{n+1}$, then the family is uniformly continuous with respect to $x\in K$.
\end{theorem}
\begin{proof}
The proof follows strictly the one of Theorem \ref{bounded_operator_2} where we replace $e_t'$ with $G_t$, and $u_{\ell}(x)$ with $e_0x_0g_{\ell,0}+e_1x_1 g_{\ell,1}+\dots +e_n x_n g_{\ell, n}$.
\end{proof}

\begin{theorem}\label{propm=5q}
Let $|a|>1$ and assume that the sequence
$$f_N(x,a):=\sum_{j=0}^N E(h_j(N) x)Z_j(N,a)$$
converges to $E(ax)$ in the space $A_{1}$.
Let $G$ be an entire left monogenic function.
Let $g_0,\dots, g_n$ be analytic on $\mathbb R$ and real valued functions. We define
$$
F_N(x,a)=\sum_{j=0}^N Z_j(N,a) G(e_0x_0g_0(h_j(N))+e_1x_1 g_1(h_{j}(N))+\dots +e_n x_n g_n(h_{j}(N))).
$$
 Then, $F_N(x,a)$ admits the supershift property that is
$$
\lim_{N\to \infty}F_N(x,a)=G(e_0x_0g_0(a)+e_1x_1 g_1(a)+\dots +e_n x_n g_n(a)),
$$
uniformly on any compact subsets of $\mathbb R^{n+1}$.
\end{theorem}
\begin{proof}
The proof follows strictly the one of Theorem \ref{new_suposs} where we replace $e_t'$ with $G_t$, and $u_{\ell}(x)$ with $e_0x_0g_{\ell,0}+e_1x_1 g_{\ell,1}+\dots +e_n x_n g_{\ell, n}$.
\end{proof}

\end{document}